\documentclass[a4paper,11pt]{article}

\usepackage{amsmath,amsthm,wasysym,amssymb,amsfonts,cite,enumerate,rotating,graphicx}

\def\d{\mathrm{d}}

\usepackage{bm} 

\newcommand{\be}{\begin{equation}}
\newcommand{\ee}{\end{equation}}
\newcommand{\bea}{\begin{eqnarray}}
\newcommand{\eea}{\end{eqnarray}}

\newcommand{\diag}{{\mathrm{diag}}}
\newcommand{\pul}{{\textstyle{\frac{1}{2}}}}
\def \bl {\mbox{\boldmath{$\ell$}}}
\def \bn {\mbox{\boldmath{$n$}}}
\def \bk {\mbox{\boldmath{$k$}}}
\def \bmu {\mbox{\boldmath{$\mu$}}}
\def \bx {\mbox{\boldmath{$x$}}}
\def \l {\ell}

\newcommand{\mbb}[1]{{{{\mbox{\boldmath{$m$}}}}_{(#1)}}}
\newcommand{\mbtt}[1]{{{{\mbox{\boldmath{$\tilde m$}}}}_{(#1)}}}

\newcommand{\Mt}[3] {{\stackrel{#1}{\tilde M}}_{{#2}{#3}}}
\def\a{\alpha}
\newcommand{\mub}[1]{{{\mathbf \mu}}}

\newtheorem{theorem}{Theorem}[section]
\newtheorem{prop}[theorem]{Proposition}
\newtheorem{corol}[theorem]{Corollary}

\theoremstyle{remark}
\newtheorem{rem}{Remark}[section]

\newcommand{\Phis}{\Phi^\mathrm{S}} 
\newcommand{\Phia}{\Phi^\mathrm{A}} 

\newcommand{\tho}{{\textrm{\thorn}}}
\newcommand{\rhob}{\bm{\rho}}

\newcommand{\Phib}{{\bm{\Phi}}}

\newcommand{\Id}{{\mathbf{1}}}

\newcommand{\beqn}{\begin{eqnarray}}
\newcommand{\eeqn}{\end{eqnarray}}
\newcommand{\pa}{\partial}
\newcommand{\ba}{\begin{array}}
\newcommand{\ea}{\end{array}}
\newcommand{\pp}{{\it pp\,}-}
\newcommand{\beqnn}{\begin{eqnarray*}}
\newcommand{\eeqnn}{\end{eqnarray*}}

\newcommand{\RT}{RT }

\def \WDS #1 {\mbox{$\Phi_{#1}^{S}$}}
\def \WDA #1 {\mbox{$\Phi_{#1}^{A}$}}
\def \WD #1 {\mbox{$\Phi_{#1}$}}
\def \kd #1 {\delta_{#1}}

\def \cL {{\cal L}}

\def \BEAH {\begin{eqnarray*}}
\def \EEAH {\end{eqnarray*}}
\def \BEA {\begin{eqnarray}}
\def \EEA {\end{eqnarray}}
\def \BDM {\begin{displaymath}}
\def \EDM {\end{displaymath}}

\def \pul {{{\footnotesize{\frac{1}{2}}}}}

\newcommand{\M}[3] {{\stackrel{#1}{M}}_{{#2}{#3}}}

\newcommand{\cM}[3] {{\stackrel{#1}{{\cal M}}}_{{#2}{#3}}}

\newcommand{\lb}{\mbox{\boldmath{$\ell$}}}

\newcommand{\mmb}[1]{{{\boldsymbol{m}}_{(#1)}}}

\newcommand{\eb}{\mathbf {e}}

\def \bbm {\mbox{\boldmath{$m$}}}

\numberwithin{equation}{section}

\title{On the Goldberg-Sachs theorem in higher dimensions \\ 
in the non-twisting case} 

\author{Marcello Ortaggio\thanks{ortaggio@math.cas.cz}, Vojt\v ech Pravda\thanks{pravda@math.cas.cz}, Alena Pravdov\'a\thanks{pravdova@math.cas.cz} \\
Institute of Mathematics, Academy of Sciences of the Czech Republic \\ \v Zitn\' a 25, 115 67 Prague 1, Czech Republic}

\begin{document}
\maketitle

\begin{abstract}

We study a generalization of the ``shear-free part'' of the Goldberg--Sachs theorem for Einstein spacetimes admitting a {\em non-twisting} multiple Weyl Aligned Null Direction (WAND) $\bl$ in $n\ge 6$ spacetime dimensions. The form of the corresponding optical matrix $\rhob$
is restricted by the algebraically special property in terms of the degeneracy of its eigenvalues. In particular, there necessarily exists at least one multiple eigenvalue and further constraints arise in various special cases. For example, when $\rhob$ is non-degenerate and the Weyl components $\Phi_{ij}$ are  non-zero, all eigenvalues of $\rhob$ coincide and such spacetimes {thus correspond} to the Robinson-Trautman (RT) class. {On the other hand, in certain degenerate cases all non-zero eigenvalues can be distinct.}  We also present explicit examples of Einstein spacetimes  admitting some of  the permitted forms of $\rhob$, including examples violating the ``optical constraint''. {The obtained restrictions on $\rhob$ are, however, in general not sufficient for $\bl$ to be a multiple WAND, as demonstrated by a few ``counterexamples''.} 
We also discuss the geometrical meaning of these restrictions in terms of integrability properties of certain null distributions.
Finally, we specialize our analysis to the six-dimensional case, {where all the permitted forms of $\rhob$ are given in terms of {just} two parameters.} 
In the appendices some examples are given and certain results pertaining to (possibly) twisting mWANDs of Einstein spacetimes are presented.

\end{abstract}

\tableofcontents

\section{Introduction}

\label{sec_intro}

In four-dimensional General Relativity, a fundamental connection between geometric optics and the algebraic structure of the Weyl tensor is provided by the Goldberg-Sachs (GS) theorem \cite{GolSac62,NP}, which states that an Einstein spacetime is algebraically special if, and only if, it contains a shearfree geodetic null congruence (cf.~\cite{Stephanibook,penrosebook2} for related results and generalizations). This theorem plays an important role in the construction of algebraically special exact solutions of the Einstein equations \cite{Stephanibook}, as the remarkable discovery of the Kerr metric shows \cite{Kerr63}. 

In recent years, the growing interest in higher dimensional gravity has motivated the study of extensions of the above concepts to  $n>4$ dimensions. An algebraic classification of the Weyl tensor based on the notion of Weyl Aligned Null Directions (WANDs) has been put forward in \cite{Coleyetal04} (see also \cite{OrtPraPra12rev} for a recent review). Furthermore, a higher dimensional version of Newman-Penrose (NP) and Geroch-Held-Penrose (GHP) formalisms have been presented in \cite{Pravdaetal04,Coleyetal04vsi,OrtPraPra07} and \cite{Durkeeetal10}, respectively. However, simple examples reveal that neither the ``geodesic part'' nor the ``shearfree part'' of the GS theorem extend in an obvious way to higher dimensions (see \cite{Ortaggioetal12} for a brief summary, along with the original references \cite{MyePer86,FroSto03,Pravdaetal04,OrtPraPra07,PraPraOrt07,GodRea09,Durkee09}).

The proper formulation of the geodesic part of the higher dimensional GS theorem has been proven in \cite{DurRea09} (see also \cite{PraPraOrt07,Durkee09}): in particular, {\em an Einstein spacetime admits a multiple WAND if, and only if, it admits a geodesic multiple WAND} -- hence one can restrict to geodesic multiple WANDs {(mWANDs)} with no loss of generality. Very recent work \cite{Ortaggioetal12} has analyzed the shearfree part in five dimensions by proving necessary conditions on the form of the optical matrix $\rhob$ (defined below in \eqref{expmatrix})  following from the existence of an mWAND.\footnote{See \cite{Taghavi-Chabert11} for a different formulation of the  ``shearfree'' condition and \cite{Ortaggioetal12} for a comparison between the two approaches.}  Contrary to the 4d case, the conditions obtained in \cite{Ortaggioetal12} are, in general, not sufficient. In fact, it seems that in higher dimensions  conditions that are both necessary and sufficient in general do not exist.

In the present contribution we focus on $n$-dimensional Einstein (including Ricci-flat) spacetimes of type~II (type D being understood as a special subcase thereof) in more than five dimensions and work out the corresponding necessary conditions on $\rhob$ { {\em under the assumption that the mWAND is twistfree} (i.e., $\rhob$ is symmetric, {equivalent to $A_{ij}=0$, cf.~\eqref{S_A} below)}}.\footnote{This automatically guarantees that $\bl$ is geodesic. Except for a few remarks, we assume {that $\rhob$ is non-zero}, since we are interested in studying its possible non-trivial forms.} 
Since the types N and III were already studied in \cite{Pravdaetal04}, we will simply connect the results of \cite{Pravdaetal04} to our analysis when appropriate. As it turns out, in $n>5$ dimensions the constraints on $\rhob$ are in general not as strong as those for $n=5$ \cite{Ortaggioetal12}.  Nevertheless, combining our results for a type II spacetime with the results for type III and type N \cite{Pravdaetal04}  gives

\begin{theorem}[Eigenvalue structure of $\rhob$ for $n\ge 6$ and $A_{ij}=0$]
\label{prop_GSHD}

In an algebraically special Einstein spacetime of dimension $n\ge 6$ that is not conformally flat, the $($symmetric$)$ optical matrix {$\rhob$} of a non-twisting multiple WAND has at least {\em one double eigenvalue}. In the following special cases stronger conditions hold and the most general permitted eigenvalue structures  are
\begin{enumerate}[(i)]
	\item if $\Phi^A_{ij} \ne 0$: $\{a,a,0,\ldots , 0\}$ {\rm\cite{Ortaggioetal12}}, 
 		\label{PhiA}
	\item if $\det \rhob\not= 0$, $\Phi_{ij}\not=0$: $\{a,a,\ldots , a\}$ (Robinson-Trautman, $\Phi_{ij}\propto \delta_{ij}$, type D(bd)), 
\label{RT}
 \item if  $\Phi_{ij}=0$ (type II(abd)): $\{a,a,b,b,c_1,\ldots , c_{n-6}\}$,   					\label{Phi0}
\item
if the type is N or III: $\{a,a,0,\ldots , 0\}$ {\rm\cite{Pravdaetal04}}. 
 \label{NIII}
 \end{enumerate}
\end{theorem}

(The Weyl components $\Phi^A_{ij}$ and $\Phi_{ij}$ are defined below in \eqref{bw0}.) The proof of (\ref{RT}) and  (\ref{Phi0}) is a new result and will be given in sections~\ref{sec_nondeg} and  \ref{sec_nontwist}, respectively.
Note that in the above cases (\ref{PhiA})--(\ref{NIII}) the matrix $\rhob$  possesses at least {\em two} double eigenvalues. In particular, in six dimensions all possible cases can be listed explicitly:

\begin{corol}[Eigenvalue structure of $\rhob$ for $n=6$ and $A_{ij}=0$]
\label{corol6d}
In a 6d algebraically special Einstein spacetime that is not confomally flat, the $($symmetric$)$ optical matrix $\rhob$ of a non-twisting multiple WAND can have only one of the following eigenvalue structures (where $a,b $ might coincide, or vanish): (i) $\{a,a,b,b \}$; (ii) $\{a,a,b,0 \}$; (iii) $\{a,b,0,0 \}$. If the spacetime is type III or N then the structure is $\{a,a,0,0\}$.

\end{corol}

Note that we do not claim that all classes compatible with theorem~\ref{prop_GSHD} or corollary  \ref{corol6d} are non-empty. 

The structure of the paper is the following. In section~\ref{sec_genII} we briefly summarize  results for general (possibly twisting) type II spacetimes in arbitrary dimension (already contained in \cite{Ortaggioetal12}), in particular the so called ``optical constraint''. These will be useful  in  the rest of the paper, where, however, we limit ourselves to the non-twisting case. 
In section~\ref{sec_nontw_gen} we set up the study of general non-twisting type II spacetimes. It turns out that it is convenient to study the case of non-degenerate and degenerate $\rhob$ separately. This is done in sections~\ref{sec_nondeg} and \ref{sec_nontwist}, where we derive certain necessary conditions for $\rhob$ being  compatible with a non-twisting mWAND, 
in particular proving thus theorem~\ref{prop_GSHD}. Using various explicit examples we show in section~\ref{sec_counter} that the necessary conditions obtained in sections~\ref{sec_nondeg} and \ref{sec_nontwist} are in general not  sufficient. In section~\ref{sec_totgeod} we elucidate the geometrical implications of theorem~\ref{prop_GSHD} in terms of existence of integrable null distributions with
totally geodesic integral null two-surfaces. In section~\ref{sec_6D} we study the six-dimensional case (corollary  \ref{corol6d}) and provide various explicit examples of such metrics.

Some related results are given in the appendices. {{In appendix~\ref{app_alg_eq} we present new algebraic constraints for general type II Einstein spacetimes. Appendix~\ref{app_type_N_III} contains a new proof of the canonical form of the optical matrix $\rhob$ for type III Einstein spacetimes in the case of a non-twisting mWAND $\bl$ (the original, longer proof was given in \cite{Pravdaetal04}). In appendix~\ref{app_shearfree} we summarize various equivalent formulations of the ``geodesic and shearfree'' condition in four dimensions, and we discuss some extensions of those to higher dimensions (including optical structure, optical constraint, integrability of certain null distributions). Appendix~\ref{app_OS5d} contains the proof (mostly relying on results of \cite{Ortaggioetal12,Wylleman_priv}) of the existence of an optical structure in all algebraically special spacetimes in five dimensions. In appendix~\ref{app_shearfreetwist} we present examples of twisting but shearfree mWANDs (in six dimensions). Finally, in appendix~\ref{app_violating}  examples of Einstein spacetimes of type D with non-degenerate mWANDs violating the optical constraint are given.}}

\paragraph{Notation}

Throughout the paper we use higher dimensional GHP formalism developed in \cite{Durkeeetal10}. We employ a null frame
\begin{equation}
  \{\lb \equiv \eb_{(0)}=\eb^{(1)},
  \bn \equiv \eb_{(1)} = \eb^{(0)},
    \mmb{i}\equiv\eb_{(i)} = \eb^{(i)} \}
\end{equation}
with indices $i,j,k,\ldots$ running from $2$ to $n-1$. The null vector fields  $\lb$ and $\bn$  and the orthonormal spacelike vector fields $\mmb{i}$  obey $\lb \cdot \bn =1$, $\mmb{i} \cdot \mmb{j} = \delta_{ij}$ and  $\lb \cdot \mmb{i} = 0 = \bn \cdot \mmb{i}.$
The optical matrix $\rhob$ is defined by
\be
 \rho_{ij} = m_{(i)}^a m_{(j)}^b \nabla_b \ell_a , \label{expmatrix}
\ee
and its trace gives the expansion scalar (up to normalization)
\be
 \rho \equiv \rho_{ii}.
\ee
{Its (anti-)symmetric parts are, respectively,
\be
 S_{ij}\equiv \rho_{(ij)} , \qquad A_{ij}\equiv \rho_{[ij]} .
 \label{S_A}
\ee 
We also define the rank of $\rhob$ as} 
\be
 m\equiv\mbox{rank}(\rhob).
\ee
Other Ricci rotation coefficients used in the paper are \cite{Pravdaetal04,OrtPraPra07,Durkeeetal10} 
\be
 L_{1i} = n^a m_{(i)}^b \nabla_b \ell_a , \quad \tau_i\equiv L_{i1} =  m_{(i)}^a n^b \nabla_b \ell_a ,\quad \M{i}{j}{0}=m_{(j)}^a\ell^b \nabla_b m_{(i)a} , \quad \M{i}{j}{k}=m_{(j)}^am_{(k)}^b \nabla_b m_{(i)a} .
 \label{L1i_M}
\ee

Boost weight (b.w.) zero components of the Weyl tensor  
\be
\Phi_{ijkl} = C_{ijkl}, \  \ \  \Phi_{ij} = C_{0i1j},\ \ \  2\Phia_{ij} = C_{01ij}, \ \ \  \Phi = C_{0101}
\label{bw0}
\ee
are subject to
\bea
&& \Phi_{ijkl} = \Phi_{[ij][kl]} = \Phi_{klij}, \qquad  \Phi_{i[jkl]}=0, \label{Phicycl} \\
&& \Phis_{ij} \equiv \Phi_{(ij)}   = -\pul\Phi_{ikjk}, \qquad \Phia_{ij} \equiv \Phi_{[ij]}, \qquad \Phi=\Phi_{ii}, \label{PhiS} 
\eea
and thus they are determined by the $\frac{1}{12}(n-1)(n-2)^2(n-3)$ components of $\Phi_{ijkl}$ and the $\pul (n-2)(n-3)$ components of $\Phia_{ij}$. In a type II spacetime, at least one of the latter two quantities must be non-vanishing when $\bl$ is an mWAND (otherwise the spacetime would be type III or more special). 
Recall that the algebraic subtypes II(a), II(b), II(c) and II(d) are defined by the vanishing of $\Phi$, the symmetric traceless part of $\Phi_{ij}$, the ``Weyl part'' of $\Phi_{ijkl}$, and $\Phia_{ij}$, respectively (see \cite{OrtPraPra12rev}).

There is no summation in case one or both repeated indices is/are  in round brackets, unless it is said otherwise. E.g., there is no summation in $\Phi_{i(j)k(j)}$ and there is summation in $\Phi_{ijkj}$ and $\sum_j \Phi_{i(j)k(j)}$.

{In what follows we shall employ the Sachs equation \cite{Pravdaetal04,OrtPraPra07,Durkeeetal10}
\begin{eqnarray}
  \tho \rho_{ij} &=& - \rho_{ik} \rho_{kj} \label{Sachs},
\end{eqnarray}
and the following  Bianchi identities (eqs.~(A10), (A11) of \cite{Durkeeetal10})    
 \begin{eqnarray}
  \tho \Phi_{ij}   &=& -(\Phi_{ik} + 2\Phia_{ik} + \Phi \delta_{ik}) \rho_{kj}, \label{A2}\label{Bi2}\\[3mm]
  -\tho \Phi_{ijkl}   &=& 4\Phia_{ij} \rho_{[kl]} - 2\Phi_{[k|i} \rho_{j|l]} + 2\Phi_{[k|j} \rho_{i|l]} 
                       + 2\Phi_{ij[k|m} \rho_{m|l]} . \label{A4}\label{Bi3}
\end{eqnarray}
In a parallelly propagated frame,  the GHP directional derivative  $\tho$ reduces to the NP derivative along the $\bl$ direction, $D=\ell^a\nabla_a$. 
For an affinely parametrized $\bl$ with an affine parameter $r$ this is simply $D=\partial_r$.}

There are two possible approaches to studying consequences of the Bianchi and Ricci equations. The first approach consists in applying the derivative operator $\tho$ on certain algebraic equations and in deriving new algebraic constraints, e.g. eq.~\eqref{thornB8}, using \eqref{Sachs}--\eqref{Bi3} (cf. also appendix~\ref{app_alg_eq} and \cite{Ortaggioetal12} (section~3.1 and appendix A therein)). In this paper we mainly use the second approach, which consists in solving the Sachs and Bianchi differential equations \eqref{Sachs}--\eqref{Bi3} and then in analyzing the compatibility of their solutions with the algebraic equations \eqref{B4}--\eqref{thornB8}.

\section{Results for general type II}
\label{sec_genII}

This section is devoted to summarizing some useful results that hold for all type II Einstein spacetimes, without assuming that $\rhob$ is non-twisting. In particular we present the algebraic restrictions on the Weyl components of b.w. zero that follow from the Bianchi identities, to be used in the next sections. More general new results are given in appendix~\ref{app_alg_eq} for future reference.
For genuine type II spacetimes the (unique) mWAND is necessarily geodesic, while for type D spacetimes there always exists a geodesic mWAND \cite{DurRea09}. In the frame used below we thus take $\bl$ to be a geodesic mWAND, without loss of generality.

\subsection{Algebraic constraints}

For type II Einstein spacetimes, the Weyl tensor must obey the algebraic constraints (A.12) and 
(B8) of \cite{Durkeeetal10} (already discussed in \cite{PraPraOrt07,Durkee09}) and $\tho$(B8) (derived in \cite{Ortaggioetal12}), namely 
\bea
2\Phia_{[jk|}\rho_{i|l]} -2\Phi_{i[j}\rho_{kl]} + \Phi_{im [jk|}\rho_{m|l]}  = 0, \label{B4} \\
\Phi_{kj} \rho_{ij} - \Phi_{jk} \rho_{ij} + \Phi_{ij} \rho_{kj} - \Phi_{ji} \rho_{jk} 
       + 2 \Phi_{ij} \rho_{jk} - \Phi_{ik} \rho  + \Phi \rho_{ik} + \Phi_{ijkl} \rho_{jl} = 0 \label{B8} , \\
\left( 2 \Phi_{kj} - \Phi_{jk} \right) \rho_{il} \rho_{jl} + \left( 2 \Phi_{ij} - \Phi_{ji} \right) \rho_{jl} \rho_{kl}  - \Phis_{ik} \rho_{jl} \rho_{jl} + \Phi \rho_{il} \rho_{kl} + \Phi_{ijkl} \rho_{js} \rho_{ls} =0. \label{thornB8}
\eea

Eq. \eqref{B8} is traceless and its
symmetric and antisymmetric parts read, respectively,
\bea
\left( 2 \Phi_{kj} - \Phi_{jk} \right) S_{ij} + \left( 2 \Phi_{ij} - \Phi_{ji} \right) S_{jk}  - \Phis_{ik} \rho + \Phi S_{ik} + \Phi_{ijkl} S_{jl} =0, \label{B8sym}\\
\Phi_{jk} A_{ji} + \Phi_{ji} A_{kj} + \Phi_{ij} \rho_{jk} -\Phi_{kj} \rho_{ji} +  \Phia_{k i} \rho + \Phi A_{ik} + \Phi_{ijkl} A_{jl} = 0 .   \label{B8asym} 
\eea

The reason for investigating algebraic conditions such as \eqref{B4}--\eqref{B8asym} (or those of \cite{Pravdaetal04} in the case of type III/N) is that they will in general constrain the possible form of $\rhob$. This way one obtains the standard Goldberg-Sachs theorem in four dimensions and one can arrive at similar conclusions also in higher dimensions, at least with additional assumptions (e.g. on the Weyl type \cite{Pravdaetal04}, on the number of dimensions \cite{Ortaggioetal12}, 
on the form of the line-element \cite{OrtPraPra09,MalPra11}, on the asymptotic behaviour \cite{OrtPraPra09b}, and/or on optical properties of $\bl$, as we shall discuss in the next sections).

\subsection{The optical constraint}

\label{subsec_OC}

It was observed in \cite{Ortaggioetal12} that the above conditions on the Weyl tensor appear to be less stringent when $\rhob$ satisfies
the {\em optical constraint} \cite{OrtPraPra09}
\be
\rho_{ik} \rho_{jk} \propto \rho_{(ij)} . \label{OC2}
\ee
In particular, when this holds eq.~\eqref{thornB8} is not an extra restriction. 
{This thus suggests that the branch of type II solutions {whose mWAND} obeys \eqref{OC2} corresponds to the case of a ``generic'' Weyl tensor \cite{Ortaggioetal12}.}
It has been proven that \eqref{OC2} indeed holds for {the Kerr-Schild vector} of all (generalized) Kerr-Schild spacetimes \cite{OrtPraPra09,MalPra11} (including Myers-Perry black holes \cite{MyePer86}), for non-degenerate geodesic double WANDs in asymptotically flat type II spacetimes \cite{OrtPraPra09b} (see eq.~(14) therein) {and for the unique double WAND of all genuine type II Einstein spacetimes in five dimensions (see footnote~\ref{foot_typeD} for the type D case) \cite{Ortaggioetal12,Wylleman_priv}.} 
Nevertheless, {double WANDs} violating the optical constraint also exist, as shown by explicit examples constructed below in section~\ref{subsubsec_examples_distinct} for $n\ge 6$, in appendix~\ref{app_violating} for $n\ge 7$, and in section~6.3 of \cite{Ortaggioetal12} for $n=5$.\footnote{We observe that all such examples are of type D. In fact, in five dimensions {\em all} type D Einstein spacetimes admitting a geodesic mWAND violating the optical constraint are known \cite{Ortaggioetal12} (and coincide with the class of Einstein spacetimes admitting a non-geodesic mWAND \cite{DurRea09}). Such ``exceptional'' null directions are necessarily twisting in 5d (but not in higher dimensions). 
However, in all such 5d type D spacetimes there always also exists a pair of non-aligned (non-twisting) mWANDs that {\em do} obey the optical constraint \cite{Ortaggioetal12}. On the other hand, this is generically not true when $n>5$ (see appendix \ref{subsubsec_ex2} for a ten-dimensional example admitting exactly two mWANDs, both violating the optical constraint).\label{foot_typeD}}
Although the above analysis applies only to type D and genuine type II  spacetimes (eqs.~(\ref{B4})--(\ref{thornB8}) become trivial identities for more special types), it is worth remarking that the optical constraint is also obeyed by the mWAND of all type N Einstein spacetimes \cite{Pravdaetal04,Durkeeetal10}. {The mWAND $\bl$ of type III Einstein spacetimes also obeys the optical constraint provided either \cite{Pravdaetal04}: (i) the spacetime is five-dimensional; (ii) the Weyl tensor satisfies a genericity condition (see~\cite{Pravdaetal04}); (iii) $\bl$ is  {\it non-twisting} (see also appendix \ref{app_type_N_III} for a simpler proof in case (iii))}.

The optical constraint implies that $(\Id- \frac{2}{\alpha} \rhob)$ is an orthogonal matrix \cite{Ortaggioetal12} and that consequently  $[\rhob, \rhob^T]$ vanishes \cite{OrtPraPra09,Ortaggioetal12}. The optical matrix $\rhob$ is therefore a {\em normal} matrix and can thus be put, using spins,  into a convenient block-diagonal form {(see \cite{OrtPraPra09,OrtPraPra10} for extended related discussions)}, i.e.,
\bea
\rhob = \alpha{\rm diag}\left(1, \dots 1, 
\frac{1}{1+ \alpha^2 b_1^2}\left[\begin {array}{cc} 1 & -\alpha b_1 \\ 
   \alpha b_1 & 1  \label{canformL} \\
  \end {array}
 \right]
, \dots, 
\frac{1}{1+ \alpha^2 b_\nu^2} \left[\begin {array}{cc} 1 & -\alpha b_\nu \\ 
   \alpha b_\nu & 1 \\ \end {array}
 \right] 
 , 0, \dots ,0
\right).
\eea
The block-diagonal form \eqref{canformL} {is useful for practical purposes because it} allows for determining the $r$-dependence of $\rhob$ by integrating the Sachs equation \eqref{Sachs} {(in a parallelly transported frame)} \cite{OrtPraPra10}. {\em In the non-twisting case it reduces to a sequence of 1s followed by a sequence of 0s} (up to an overall factor). Note that the symmetric part of each two-block is proportional to a two-dimensional identity matrix, i.e. it is ``shear-free''. In four dimensions the optical constraint implies that either there is a single such block (in which case $\rho_{ij}$ is shearfree) or that $\rho_{ij}$ is symmetric with exactly one non-vanishing eigenvalue. However, the Goldberg-Sachs theorem shows that the latter case cannot occur. Therefore in 4d the optical constraint is a necessary condition for $\bl$ to be a repeated principal null direction but it is not sufficient \cite{Ortaggioetal12}.

On the other hand, in higher dimensions a vector field $\bl$ obeying the optical constraint is in general {\em shearing}, as in the case of Myers-Perry black holes \cite{MyePer86} (see also \cite{FroSto03,PraPraOrt07} for a discussion of their optical properties). Together with other results \cite{Pravdaetal04,PodOrt06,OrtPraPra07}, this has made clear that the shearfree condition is in general ``too restrictive'' in higher dimensions, as opposed to the four-dimensional case. In particular, it was observed in \cite{OrtPraPra07} that in {\em odd} dimensions a twisting geodesic mWAND is necessarily shearing. By contrast, twisting geodesic mWANDs with zero shear are permitted in {\em even} dimensions and they have necessarily $\det(\rhob)\neq 0$ (as can be easily seen in a frame adapted to $A_{ij}$, using the fact that $S_{ij}\propto\delta_{ij}$). An explicit example in six dimensions is discussed in appendix \ref{app_shearfreetwist}.

In the non-twisting case, the existence of shearfree spacetimes has been already known for some time in all dimensions -- they are either \RT \cite{PodOrt06} or Kundt \cite{PodZof09} solutions, according to the presence/absence of expansion.

\subsection{Possible generalizations of the geodesic{\&}shearfree property}

In arbitrary dimensions various geometric conditions can be considered which are different from the standard geodesic{\&}shearfree condition (considered above) and which, however, become all equivalent (except for the optical constraint, cf.~\cite{Ortaggioetal12}) in the special case $n=4$.  Further such conditions are discussed in appendix~\ref{app_shearfree}. 
One could thus conceive that various formulations of the Goldberg-Sachs theorem are in principle possible when $n>4$. Some of these formulations have been studied in \cite{Taghavi-Chabert11,Taghavi-Chabert11b,Ortaggioetal12} (see also \cite{OrtPraPra09}) but none of these gives necessary and sufficient conditions. In the rest of this paper we will discuss necessary conditions determined by the presence of a non-twisting mWAND in an $n$-dimensional Einstein spacetime and we will present several explicit examples. A possible interpretation of these results in terms of the geometric conditions of appendix~\ref{app_shearfree} will be discussed in sections~\ref{sec_totgeod} (for $n\ge6$) and \ref{subsec_integrability} (for $n=6$).

\section{Non-twisting $\bl$: general properties}
\label{sec_nontw_gen}

Here we study the optics of a hypersurface orthogonal mWAND $\bl$ in a type II Einstein spacetime.  This mWAND is automatically geodetic. Thus we have 
\be
 \kappa_i=0=A_{ij} ,
\ee
{so that $\rho_{ij}=S_{ij}$.} Since the algebraic equations \eqref{B4}--\eqref{B8asym}  are trivial for  Kundt spacetimes, in what follows we assume $\rhob\not= ${\bf 0}.

\subsection{Case $\Phia_{ij} \neq 0$}

\label{sec_PhiA}

This case has been already analyzed in \cite{Ortaggioetal12}, {arriving at} point (i) of theorem~\ref{prop_GSHD} of section~\ref{sec_intro}. As already observed there, in this case $\rhob$ satisfies the optical constraint and it is  necessarily degenerate ($m=2$).  Moreover it is shearing, except for $\rhob=0$ {(which is necessarily the case for $n=4$)}.

It remains to consider the case when $\Phi^A_{ij}=0$.

\subsection{Case $\Phia_{ij}=0$}

This case defines the subtype II(d) {in the notation of \cite{Coleyetal04}}. For  $\Phia_{ij}=0$, eq.~\eqref{B8} reduces to  (recall $A_{ij}=0$ here)
\be
 -\rho \Phi_{ik} + \Phi \rho_{ik} + 2\Phi_{ij} \rho_{jk} + \Phi_{ijkl} \rho_{jl}=0 ,  \label{eqn:algI}
\ee
and eq. \eqref{B8asym}  to 
\be
[\rhob,\Phib^S]=0.\label{comm_rho_Phi}
\ee
Note that taking into account  (\ref{comm_rho_Phi}), eq. (\ref{eqn:algI}) is symmetric and corresponds to eq.~\eqref{B8sym}.

Similarly, eq. \eqref{thornB8} yields 
\be
 -\Phi_{ik} \rho_{jl} \rho_{jl} + \Phi (\rho^2)_{ik} + 2 \Phi_{ij} (\rho^2)_{jk} + \Phi_{ijkl} (\rho^2)_{jl} = 0\label{eqn:rhoC}.
 \ee

Thanks to \eqref{comm_rho_Phi}, one can choose a {basis where both $\rho_{ij}$ and $\Phi_{ij}$ are diagonal}, $\rho_{ij}=\diag(\rho_2,\rho_3,\dots)$,  
$\Phi_{ij}=\diag(\Phi_2,\Phi_3,\dots)$, { therefore $\Phi_{ijkj}=0$ for $i\neq k$}.
In this frame, {the off-diagonal components of the algebraic constraint~(\ref{eqn:algI}) are} 
\be 
 \sum_j\rho_{(j)}\Phi_{i(j)k(j)}=0 \qquad (\mbox{for } k\neq i) .
 \label{nondiag_constr}
\ee
 The diagonal part of~(\ref{eqn:algI}) can be expressed as  ${\cal L}_{ij} \tilde \rho_j =0$, 
where $\tilde \rho_{i}$ is the  vector  $\tilde \rho_{i}=(\rho_2, \rho_3 \dots)$  
and the linear operator ${\cal L}_{ij}$ is given by
\[
  {\cal L}_{ij} = W_{ij} + \Phi  \delta_{ij} + 2 \diag (\Phi_{2},\Phi_{3},\dots \Phi_{n}) - \left[ 
                     \begin{array}{cccc}
                       \Phi_{2}                  & \Phi_{2}        & \dots & \Phi_{2}     \\
                       \Phi_{3}       & \Phi_{3}       & \dots & \Phi_{3}   \\
                       \vdots    & \vdots    & \dots & \vdots\\
                       \Phi_{n} & \Phi_{n} & \dots & \Phi_{n}
                     \end{array}
               \right],
\]
or, equivalently, ${\cal L}_{ij} = W_{ij} +( \Phi + 2\Phi_{(i)} )  \delta_{ij} -\Phi_{(i)} $, where {we have defined the symmetric and traceless matrix}
\be
	W_{ij} \equiv C_{(i)(j)(i)(j)} \qquad {\mathrm{\ \ (no\ summation)}}.\label{def_W}
\ee
{Note that $W_{(i)(i)}=0$.} From  the first of \eqref{PhiS}, it follows  
\be
\sum_j W_{ij} = - 2 \Phi_i ,
 \label{W_phi}
\ee
{which will be used in certain calculations throughout the paper.} (From  \eqref{W_phi} 
it follows that $W_{ij}=0\Rightarrow {\cal L}_{ij}=0$.)  
Thus the sum of the rows of ${\cal L}_{ij}$ vanishes and so they are  linearly dependent and 
\be
\det {\cal {\bm{L}}} =0. \label{detL}
\ee 
Therefore  zero is an eigenvalue of ${\cal L}$ and  (the diagonal part of) (\ref{eqn:algI})  admits non-trivial solutions.

Note that (\ref{eqn:rhoC}) has the form ${\cal L}_{ij} {\tilde \rho}^2_j =0$, where components of  the $(n-2)$ dimensional vector ${\tilde \rho}^2$ are squares of components of ${\tilde \rho}$.  The characteristic polynomial of ${\cal L}$ is 
\be
 \det({\cal L}-\lambda I)=\lambda^{n-2} + k_{n-3} \lambda^{n-3} + \dots + k_1 \lambda + k_0 ,
\ee
with $k_0=0$ due to (\ref{detL}).
Now zero is a multiple eigenvalue of ${\cal L}$ iff  $k_1=0$. Let us observe that for a generic form of a type II Weyl tensor $k_1$ is non-vanishing {and thus zero is a single eigenvalue 
of ${\cal L}$}.\footnote{For example in five dimensions $k_1 = 12 \left(\phi_2 \phi_4+ \phi_3 \phi_4 + \phi_2 \phi_3 \right)$. {If we consider, for instance, five-dimensional} \RT spacetimes ({which coincide with the Schwarzschild solution plus a possible cosmological constant}), characterized by $\rho_{ij} \propto \delta_{ij}$,  equation \eqref{eqn:algI} implies $\Phi_{ij} \propto \delta_{ij}$, which clearly leads to $k_1 \not=0$.} 
If this is the case then   ${\tilde \rho}^2$ is proportional to ${\tilde \rho}$ and therefore 
\be
\rho_{ij}= \alpha \diag(1,1, \dots 1,0,\dots 0).
\ee
Recalling also point (i) of theorem~\ref{prop_GSHD},
we can conclude with 

\begin{prop}
For non-twisting type II Einstein spacetimes with $k_1 \not=0$, all non-vanishing  eigenvalues of the optical matrix $\rho_{ij}$ coincide. 
\label{prop_nontwist}
\end{prop}

If $\Phia_{ij} \not= 0$, the stronger result of point (i) of theorem~\ref{prop_GSHD} holds.

As we will see in sections~\ref{subsec_alldistinct} and \ref{subsec_permitted_6D} (table~\ref{tab_6D}), 
there do exist non-twisting type II Einstein spacetimes with distinct non-vanishing eigenvalues of $\rhob$. 
According to Proposition~\ref{prop_nontwist}, for these spacetimes $k_1$ vanishes. In particular $\cal L$, and therefore also $k_1$, vanishes for spacetimes with $W_{ij}=0$. We will study special cases with various parts of the Weyl tensor vanishing in the following sections.

\section{Non-twisting $\bl$: non-degenerate $\rhob$ ($m=n-2$)}

\label{sec_nondeg}

First, let us consider the case of a non-degenerate $\rhob$ (i.e., $\det\rhob\not= 0$ -- this is relevant, e.g., for asymptotically flat algebraically special spacetimes \cite{OrtPraPra09b}). When $\det\rhob\not= 0$ one necessarily has (see point (i) of theorem~\ref{prop_GSHD}) 
\be
 \Phia_{ij}=0 ,
\ee
so that the type is II(d). As noticed above we can use a frame in which both $\rho_{ij}$ and $\Phi_{ij}$ are diagonal, which is moreover compatible with parallel transport \cite{PraPra08,OrtPraPra10}. Thus the eigenvalues of the optical matrix  are \cite{PraPra08}
\be
 \rho_i=\frac{1}{r-b_i} .
 \label{rho_i}
\ee

{In order to discuss the possible structures of $\rhob$ it is convenient to discuss separately various cases in which different parts of the Weyl tensor are (non-)zero.}

\subsection{Cases $\Phi_{ij}\neq0$ and $\Phi_{ij}=0\neq W_{ij}$}

\label{subsec_nondeg_Phi}

From Bianchi equation~\eqref{Bi2}  one has
\beqn
 D\left(\frac{\Phi_{(i)}}{\rho_{(i)}}\right)=-\Phi .
 \label{nondeg_DPhi_i}
\eeqn
The case when all  $\rho_i$ coincide (and are non-zero) is the known \RT case, for which also all  $\Phi_i$ must coincide {\cite{PraPra08,PodOrt06}.}
{Thus} let us  consider here the case when at least two eigenvalues $\rho_i$ are different.

Since the r.h.s. {of~(\ref{nondeg_DPhi_i})} is the same for any $i$, we obtain that either all $\Phi_i=0$, or
\be
 \Phi_{i}=\rho_i(A+\Phi^0_{(i)}) ,
 \label{Phi_i}
\ee
where $A=A(r)$ is a function that must satisfy the equation $DA+A\rho=-\Phi^0_{j}\rho_j$ (since $\Phi=A\rho+\Phi^0_{j}\rho_j$) and, from now on, quantities with a superscript $^0$ do not depend on $r$. Using the intermediate substitution $A(r)=Y(r)\prod_k\rho_k$, one arrives at the solution 
\be
 A=\left(\prod_k\frac{1}{r-b_k}\right)\left[-\sum_j\Phi^0_{j}\int\d r\prod_{l\neq j}(r-b_l)+A^0\right] .\label{A_nondeg}
\ee
For subsequent discussions it is useful to rewrite \eqref{A_nondeg} using partial fraction decomposition, i.e., 
\be
A=\frac{P_{n-2}}{\prod_k {(r-b_k)}}=a^0+\frac{P_{n-3}}{\prod_k {(r-b_k)}}=a^0+
\sum_{K=1}^{\alpha_{max}} \sum_{\forall i,  \alpha_i\geq K}\frac{c^{(K)}_{(i)}}{(r-b_i)^K} ,\label{A_decomposed}
\ee
where $P_k$ denotes a polynomial of order $k$ in $r$,  $c^{(K)}_{(i)}$ {do not depend on $r$} and $\alpha_i$ denotes the multiplicity of $b_i$, with $\alpha_{max}$ being the maximal multiplicity. {In particular, the term $a^0$ can be determined by looking at} the leading term of \eqref{A_nondeg} in the limit $r\to\infty$, i.e., 
\be
 A=-\frac{1}{n-2}\sum_k\Phi^0_k+O(r^{-1}) ,
\ee
so that
\be
a^0=-\frac{1}{n-2}\sum_k\Phi^0_k.\label{am}
\ee
In the case $\Phi^0_j=0$, $a^0$ vanishes and  one has  simply $A=A^0/\prod_k {(r-b_k)}$.

Next, the equation for $W_{ij}$ (see \eqref{Bi3}) can be written as (but recall $W_{(i)(i)}=0$)
\beqn
 & & D\left(\frac{W_{ij}}{\rho_{(i)}-\rho_{(j)}}\right)=\frac{\Phi_{i}\rho_j+\Phi_{j}\rho_i}{\rho_{(i)}-\rho_{(j)}} \qquad (\rho_{j}\neq\rho_{i}) , \\
 & & D\left(\frac{W_{ij}}{\rho_{(i)}^2}\right)=\frac{\Phi_{i}+\Phi_{j}}{\rho_{(i)}} \qquad (\rho_{j}=\rho_{i}) .
\eeqn
Using the above expression for $\Phi_i$, in both cases one can write the solution as
\be
 W_{ij}=\frac{1}{(r-b_{(i)})(r-b_{(j)})}\left[2\int\d r A+(\Phi^0_{i}+\Phi^0_{j})r+ W^0_{ij}\right] .\label{Wij_nondeg}
\ee

Now, imposing $-2\Phi_i=\sum_jW_{ij}$ {(recall~(\ref{W_phi}))} we obtain the constraint
\be
 -2(A+\Phi^0_{i})=\sum_{j\neq i}\frac{1}{r-b_{(j)}}\left[2\int\d r A+(\Phi^0_{i}+\Phi^0_{j})r+ W^0_{ij}\right] .
 \label{constrW}
\ee
{For $r\to\infty$, using \eqref{A_decomposed} and \eqref{am} at 
the leading order (\ref{constrW}) implies that all  $\Phi^0_{i}$ coincide, i.e., }
\be
 \Phi^0_{i}=\frac{1}{n-2}\sum_k\Phi^0_k\equiv{f^0}. 
 \label{Phi_i0}
\ee
$A$ thus becomes
\be
 A=-f^0+A^0\prod_k\frac{1}{r-b_k} ,
 \label{A_2}
\ee
where we have used $\sum_j\prod_{l\neq j}(r-b_l)=D[\prod_k(r-b_k)]$, so that
\be
 \Phi_{i}=\frac{A^0}{r-b_{i}}\prod_k\frac{1}{r-b_k} , \qquad W_{ij}=\frac{1}{(r-b_{(i)})(r-b_{(j)})}\left[2A^0r\prod_k\frac{1}{r-b_k}+ W^0_{ij}\right] .
 \label{nondeg_Weyl_simplif}
\ee

The constraint~(\ref{constrW}) can thus be written as
\be
 -2A^0\prod_k\frac{1}{r-b_k}=\sum_{j\neq i}\frac{1}{r-b_{(j)}}\left[2A^0\int\d r\prod_k\frac{1}{r-b_k}+W^0_{ij}\right] \qquad (i=2,\ldots,n-1) .
 \label{constrW2}
\ee

{It is now useful to discuss separately various subcases with different multiplicity of eigenvalues.} 

\subsubsection{All the $b_i$ coincide: shearfree congruences (Robinson-Trautman spacetimes)}

In the \RT case all the $b_i$ {of (\ref{rho_i})} coincide {(and can be set to zero by shifting $r$, if desired)} and {from~(\ref{constrW2})} we simply obtain $\sum_{j}W^0_{ij}=0$. 

On the other hand, let us assume in the following that not all $b_i$ coincide, i.e., at least two of these are distinct, say $b_2\neq b_3$. 

\subsubsection{Shearing case with all $b_i$ distinct: not permitted}

First, if all  $b_i$ are distinct, {i.e., $\alpha_{max}=1$ in \eqref{A_decomposed}}, we can compute explicitly the required integral using partial fraction decomposition 
\be
 \int\d r\prod_k\frac{1}{r-b_k}=\sum_k\frac{\ln(r-b_k)}{\prod_{l\neq k}(b_k-b_l)} \qquad \mbox{($b_k$ all distinct)}.
\ee
It is thus clear that the singularity structure of the l.h.s. and the r.h.s. of (\ref{constrW2}) (for any $i$) cannot be the same unless $A^0=0$. Therefore, {from~(\ref{constrW2})} also $\sum_{j}(r-b_{(j)})^{-1}W^0_{ij}=0$. However, this condition  implies that  $b_i$ cannot be all distinct, so that this case is in fact not permitted (unless $W^0_{ij}=0$, so that $W_{ij}=0$ and therefore also $\Phi_{ij}=0$, contrary to our assumptions here -- however, we will see in sections~\ref{subsec_nondeg_Phi3} and \ref{subsec_nondeg_Phi4}  below that the case with all distinct $b_i$ is ruled out also for $W^0_{ij}=0$).

\subsubsection{Shearing case with at least one $b_i$ repeated}
\label{subsec_onerep}

Let us now consider the case when at least one $b_i$ is repeated, say $b_2$, with multiplicity $1<\alpha_2<n-2$, and there exists $b_3\neq b_2$ (this is not a restriction since we are excluding here the \RT case, which corresponds to $\alpha_2=n-2$). 

Using partial fraction decomposition {similarly as in}~\eqref{A_decomposed}, the left and right hand sides {of (\ref{constrW2}) take} the form (for $i=3$)
\bea
\mbox{lhs}&=&
-2A^0\left[\frac{p_{\alpha_2}}{(r-b_2)^{\alpha_2}}+(\mbox{terms with lower order poles at } b_2 \mbox{ and with no poles at }b_2)\right],\nonumber\\
\mbox{rhs}&=&
 -2A^0\frac{p_{\alpha_2}\alpha_2}{(\alpha_2-1)(r-b_2)^{\alpha_2}} +(\mbox{terms with lower order poles at } b_2 \mbox{ and with no poles at }b_2) , \nonumber
\eea
respectively, where $p_{\alpha_2}\not= 0$. Comparing the highest order terms in $1/(r-b_2)$ we thus  get
\be
 A^0=0, \qquad \sum_{j}\frac{W^0_{ij}}{r-b_{(j)}}=0 .
\ee

In particular, it follows from (\ref{nondeg_Weyl_simplif}) that 
\be 
 \Phi_i=0 ,
\ee
so that the Weyl type is II(abd), and \eqref{Wij_nondeg} thus gives 
\be
 W_{ij}=\frac{W^0_{ij}}{(r-b_i)(r-b_j)} ,
 \label{Wij}
\ee
with $W^0_{ij}=W^0_{ji}$ and $W^0_{(i)(i)}=0$.

Now, because of (\ref{Wij}) the condition $\sum_jW_{ij}=0$ constraints the possible multiplicities of  $b_i$ (i.e., of the eigenvalues of $\rho_{ij}$). First, if $b_i$ are all distinct, we immediately get $W_{ij}^0=0$, {as already discussed above}. Similarly, recalling $W_{ij}^0=W_{ji}^0$ it is easy to see that for $W_{ij}^0\neq0$ the structures  $\{a,a,c_1,\ldots,c_{n-4}\}$ and  
 $\{a,a,a,c_1,\ldots,c_{n-5}\}$ are also forbidden.\footnote{To see that 
 $\{a,a,a,c_2,\ldots,c_{n-5}\}$ cannot occur, assume $b_2=b_3=b_4$ and all remaining eigenvalues are single. Then $\sum_jW_{ij}=0$ with (\ref{Wij}) gives
$W^0_{i2}+W^0_{i3}+W^0_{i4}=0$ and $W^0_{i\mu}=0$ ($\mu\neq 2,3,4$). Since $W^0_{ij}=W^0_{ji}$, this implies that the only possible non-zero components are $W_{23}$, $W_{24}$ and $W_{34}$, with the conditions $W_{23}^0+W_{24}^0=0$, $W_{32}^0+W_{34}^0=0$ and $W_{42}^0+W_{43}^0=0$. However, this system admits only the solution $W_{23}^0=W_{24}^0=W_{34}^0=0$, so that $W_{ij}=0$.\label{foot_311}}  
However, the structure  
$\{a,a,b,b,c_1,\ldots,c_{n-6}\}$ is compatible with such constraints (take, e.g., $b_2=b_3$, $b_4=b_5$ and $W_{34}^0=-W_{35}^0=W_{25}^0=-W_{24}^0\neq0$, {and the remaining $W_{ij}$ equal to zero}). 

To summarize, we have seen above that in type II Einstein spacetimes with a non-degenerate non-twisting mWAND $\bl$, one has $\Phia_{ij}=0$ and
\begin{enumerate}
 \item if $\Phi_{ij}\neq0$, then $\bl$ must be shearfree and the corresponding spacetimes belong to the \RT class. {This includes, in particular, the result of \cite{OrtPraPra09b} for asymptotically flat type II vacuum spacetimes (restricted to non-twisting case)}\footnote{Ref.~\cite{OrtPraPra09b} used an expansion method. Instead of the condition $\Phi_{ij}\neq 0$, a condition on the asymptotic fall-off behaviour of the Weyl tensor was assumed there. In the present notation, that amounts to taking $W_{ij}^0=\Phi^{\{3\}}_{ijkl}=\Phi^{\{4\}}_{ijkl}=0$, which requires that $\Phi_{ij}\neq 0$ (otherwise the type would be III). The assumptions of \cite{OrtPraPra09b} in the non-twisting case were thus stronger than those used here (note indeed that \RT spacetimes with $W_{ij}^0\neq 0$ do exist \cite{PodOrt06}). {(The symbols $\Phi^{\{3\}}_{ijkl}$ and $\Phi^{\{4\}}_{ijkl}$ are defined in sections~\ref{subsec_nondeg_Phi3} and \ref{subsec_nondeg_Phi4} below.)}} 
\item if $\Phi_{ij}=0$ (type II(abd)) the structure of eigenvalues of $\rho_{ij}$ can be more generic, however for $\Phi_{ij}=0\neq W_{ij}$ it must have the multiplicities $\{a,a,b,b,c_1,\ldots,c_{n-6}\}$ (or more special). 
\end{enumerate}

Let us now discuss the remaining cases (in which, without loss of generality, we could assume $\Phi_{ij}=0=W_{ij}$ --  however this will actually not be used to prove the following results).

\subsection{Case $\Phi^{\{3\}}_{ijkl}\neq0$} 

\label{subsec_nondeg_Phi3}

{Here we assume that a component of $\Phi_{ijkl}$ with precisely three distinct values of $i,j,k,l$ is non-vanishing,
i.e., for some $i\not=k$, we have $\Phi_{i(j)k(j)} \not=0$.} This case is possible {only for $n\ge6$} because we need at least three possible distinct values for $i,j,\ldots$ and because of the tracefree condition $\Phi_{ijik}=0$ with $k\neq j$ (however, {for $n=4,5$ necessarily $\Phi_{ij}\neq0$} and from the above discussion one already knows that \RT is the only possibility).

Eq.~\eqref{B4}  gives
\be
 \Phi_{(i)(j)(i)(k)}(\rho_{(j)}-\rho_{(k)})=0 \qquad (k\neq j) .\label{B4offdiag}
\ee
Without loss of generality, we can assume $\Phi_{2324}\neq0$, so that necessarily $b_3=b_4$. 

Now, the off-diagonal ($i\not=k$) components of \eqref{Bi3} read
\be
-D \Phi_{i(j)k(j)} = \Phi_{i(j)k(j)} \rho_{(j)} + \Phi_{i(j)k(j)} \rho_{(k)}.
\ee
Taking into account $\rho_k = (r-b_k)^{-1}$, we arrive at
\be
\Phi_{i(j)k(j)} = \frac{\Phi^{0}_{i(j)k(j)}}{(r-b_{(j)})(r-b_{(k)})}  \qquad (i \not= k). \label{Phi_ijkj}
\ee

Then,    
the condition $\Phi_{ijkj}=0$, $k\not=i$, implies that at least another component $\Phi_{(i)3(i)4}$ must be non-zero, say $\Phi_{5354}\neq0$, and that $b_5=b_2$. The eigenvalue structure must therefore again be  
$\{a,a,b,b,c_1,\ldots,c_{n-6}\}$ (or more special).

\subsection{Case $\Phi^{\{4\}}_{ijkl}\neq0$} 

\label{subsec_nondeg_Phi4}

Here we assume that there exists a non-vanishing component of $\Phi_{ijkl}$ with all values of $i, j, k, l$ being distinct. This case is possible only for $n\ge6$. 

Similarly as in the previous case, eq.~\eqref{Bi3} implies
\be
 \Phi_{ijkl} = \frac{\Phi^{0}_{ijkl}}{(r-b_{(k)})(r-b_{(l)})},  \qquad  {\mathrm{\ (for \ }} i,j,k,l {\mathrm{\ all \ distinct}}). 
\ee
However, since $\Phi_{ijkl}=\Phi_{klij}$ we obtain $b_i=b_k$ and $b_j=b_l$ (or $b_i=b_l$ and $b_j=b_k$) and the eigenvalue structure is again 
$\{a,a,b,b,c_1,\ldots,c_{n-6}\}$ (or more special). 

\subsection{Summary}

No further cases are possible, since $\Phi^{\{4\}}_{ijkl}=\Phi^{\{3\}}_{ijkl}=W_{ij}=0$ implies the type III. 
To summarize, {in sections~\ref{subsec_nondeg_Phi}--\ref{subsec_nondeg_Phi4}} we have thus shown that
\begin{prop}

\label{prop_non_deg2}

For type II  Einstein spacetimes, {the existence of} a non-twisting, non-degenerate double WAND implies that the algebraic type is necessarily II(d), i.e. $\Phia_{ij}=0$, or more special, and

\begin{enumerate}

\item if $\Phi_{ij}\neq 0$ the spacetime is shearfree (Robinson-Trautman), i.e. the eigenvalue structure of $\rho_{ij}$ is 
 $\{a,\ldots,a\}$ {(with $a\neq0$)}. This is {the only possibility when} $n=4,5$.\footnote{Since type II(abd) coincides with type III in those dimensions.} 
One has $\Phi_{ij}=A^0\delta_{ij}/(r-b_0)^{n-1}$, so that the type is D(bd) 
(one can set $b_0=0$ by a shift of the affine parameter $r$).

\item if $\Phi_{ij}=0$ {(type II(abd))} the structure is  
$\{a,a,b,b,c_1,\ldots,c_{n-6}\}$ (or more special; $a,b,c_\alpha\neq0$). In particular, in six dimensions this means 
$\{a,a,b,b\}$ (see also section~\ref{sec_6D}).

\end{enumerate}
\end{prop}
This is nothing but a more detailed version of points (\ref{RT}) and (in the subcase $\det\rhob\neq0$) (\ref{Phi0}) of theorem~\ref{prop_GSHD} of section~\ref{sec_intro}, which are thus proven (the proof of  (\ref{Phi0}) will be completed in section~\ref{sec_nontwist} for any value of rank($\rhob$), cf. remark~\ref{rem_Phi0}).
In particular, (for $n>5$) there are always at least two double eigenvalues. 
In this sense, we will see that the situation is different in the degenerate case. {Note that spacetimes of point 1. (all explicitly known \cite{PodOrt06}) obey the optical constraint, whereas those of point 2. do not, in general. Examples of the latter in $n\ge7$ dimensions will be provided in appendix~\ref{app_violating}.}

\section{Non-twisting $\bl$: degenerate $\rhob$ ($0<m<n-2$) }

\label{sec_nontwist}

Let $m$ denote the rank of $\rho_{ij}$. The value $m=n-2$ corresponds to the previously considered non-degenerate case, while $m=0$ defines Kundt. Therefore here we shall restrict to $0<m<n-2$ .
We need to define two types of indices to distinguish between non-vanishing and vanishing eigenvalues, namely
\beqn
 & & \rho_p=\frac{1}{r-b_p} \qquad (p,q,o,t=2,\ldots, m+1) , \\
 & & \rho_z=0 \qquad (z,v,w,y=m+2,\ldots, n-1) . \label{rho_z_0}
\eeqn

Recall that (point (i) of theorem~\ref{prop_GSHD}) 
in the non-twisting case with $\Phi^A_{ij}\neq0$ one has $m=0,2$ and $\rho_{ij}=\mbox{diag}(a ,a ,0,\ldots,0)$ for any $n>4$, {so that this case does not require further investigation}. 
In all remaining cases we thus have 
\be
 \Phi^A_{ij}=0 \qquad (m\neq 0,2).
\ee
{In the following we will give the $r$-dependence of the Weyl components, which is then used to constraint the possible forms of $\rhob$. In particular, we shall explore under what conditions  all the eigenvalues of $\rhob$ {can} be distinct (which is not permitted in the non-degenerate case). We shall also discuss some special cases and construct explicit examples.}

\subsection{$\Phi_{ij}$ and $W_{ij}$ components}

\label{subsec_Phi_W_deg}

Proceeding similarly as in section~\ref{sec_nondeg} but {(thanks to (\ref{rho_z_0}))} with the additional equation $D\Phi_z=0$, one finds
\beqn
 & & \Phi_{p}=\rho_p(A+\Phi^0_{(p)}) , \qquad \Phi_{z}=\Phi_{z}^0 , \\
 & & W_{pq}=\frac{1}{(r-b_{(p)})(r-b_{(q)})}\left[2\int\d r A+(\Phi^0_{p}+\Phi^0_{q})r+ W^0_{pq}\right] , \\
 & & W_{pz}=\frac{1}{r-b_{(p)}}(\Phi^0_zr+W_{pz}^0) , \qquad W_{zv}=W_{zv}^0 ,
\eeqn
where 
\be
 A=\left(\prod_o\frac{1}{r-b_o}\right)\left[-\sum_p\Phi^0_{p}\int\d r\prod_{q\neq p}(r-b_q)-\sum_z\Phi^0_z\int\d r\prod_{q}(r-b_q)+A^0\right] .
 \label{A_degen}
\ee

Now, imposing $-2\Phi_z=\sum_pW_{zp}+\sum_vW_{zv}$ we obtain the constraints
\be
 \sum_vW_{zv}^0=-(m+2)\Phi^0_z , \qquad \sum_p\frac{\Phi^0_zb_p+W^0_{pz}}{r-b_p}=0 .
 \label{constr_Phi_z}
\ee

Next, imposing $-2\Phi_q=\sum_pW_{qp}+\sum_zW_{qz}$ gives 
\be
 -2(A+\Phi^0_{q})=\sum_{p\neq q}\frac{1}{r-b_{(p)}}\left[2\int\d r A+(\Phi^0_{q}+\Phi^0_{p})r+ W^0_{pq}\right]+\sum_z(\Phi^0_zr+W^0_{qz}) .
 \label{constrW_deg}
\ee
By comparing the leading terms of the l.h.s. and of the r.h.s. for $r\to\infty$ one obtains 
\be
 m\Phi^0_{q}+\frac{1}{m+1}b_q\sum_z\Phi^0_z+\sum_zW^0_{qz}=\sum_p\Phi^0_{p}-\frac{1}{m+1}\sum_pb_p\sum_z\Phi^0_z .
 \label{constrW_deg_leading}
\ee

\subsection{$\Phi^{\{3\}}_{ijkl}$ and $\Phi^{\{4\}}_{ijkl}$ components} 

\label{subsec_Phi3_4_deg}

To complete the description of the Weyl tensor, we now give the general form of the $\Phi^{\{3\}}_{ijkl}$ and $\Phi^{\{4\}}_{ijkl}$ components (recall that these are non-zero only for $n\ge6$). In particular, we also discuss constraints following from the assumption that all  non-vanishing eigenvalues of $\rho_{ij}$ are distinct (useful for later analysis). 

\subsubsection{$\Phi^{\{3\}}_{ijkl}$ components}
\label{subsubPhi3deg} 

Eq.~\eqref{B4}  gives
\be
 \Phi_{(o)(p)(o)(q)}(\rho_{(q)}-\rho_{(p)})=0=\Phi_{(z)(p)(z)(q)}(\rho_{(q)}-\rho_{(p)}) , \qquad \Phi_{(q)p(q)z}=0=\Phi_{(v)p(v)z} , \quad (p\neq q, \ v\neq z) .
 \label{A12_deg_Phi3}
\ee
Non-vanishing components must satisfy the tracefree conditions (recall that $\Phi_{ij}$ is diagonal in the frame we are using)
\be
 \Phi_{poqo}+\Phi_{pzqz}=0 , \qquad \Phi_{zpvp}+\Phi_{zwvw}=0 \qquad (p\neq q, \ v\neq z) .
 \label{trace_deg_Phi3}
\ee

The $r$-dependence, following from~\eqref{Bi3}, is
\beqn
 & & \Phi_{(p)z(p)v}=\frac{\Phi_{(p)z(p)v}^0}{r-b_{(p)}} \quad (v\neq z), \qquad \Phi_{(z)v(z)w}={\Phi_{(z)v(z)w}^0} \quad (v\neq w) , \nonumber \\
 & & \Phi_{(p)q(p)o}=\frac{\Phi_{(p)q(p)(o)}^0}{(r-b_{(p)})(r-b_{(o)})} ,  \quad (q\neq o) \qquad  \Phi_{(z)p(z)q}=\frac{\Phi_{(z)p(z)(q)}^0}{r-b_{(q)}}  \quad (p\neq q) .
 \label{r_deg_Phi3}
\eeqn
From~(\ref{trace_deg_Phi3}), the quantities in~(\ref{r_deg_Phi3}) must satisfy the constraints
\beqn
 & & \Phi_{pzqz}^0=0 , \qquad \sum_o\frac{\Phi^0_{p(o)q(o)}}{r-b_{(o)}}=0 \qquad (p\neq q) , \label{r_deg_Phi3_constr_0}\\
 & & \Phi^0_{zwvw}=0 , \qquad \sum_p\frac{\Phi_{z(p)v(p)}^0}{r-b_{(p)}}=0 \qquad (v\neq z) . \label{r_deg_Phi3_constr}
\eeqn

Note that, by the symmetries of the Weyl tensor, {from} the {third} equation of~(\ref{r_deg_Phi3}) it follows that 
$\Phi_{(p)q(p)o}\neq0\Rightarrow b_q=b_o$ (for $q\neq o$), in agreement with~(\ref{A12_deg_Phi3}); using also the {second} equation of~\eqref{r_deg_Phi3_constr_0} we get
\be
 \Phi^{\{3\}}_{(p)q(p)o}\neq0 \Rightarrow \rhob=\diag(a,a,b,b,c_1\dots) \qquad (a,b\neq0). 
 \label{Phi3_new}
\ee
Further, {the fourth} equation of~(\ref{r_deg_Phi3}) gives $\Phi_{(z)p(z)q}\neq0\Rightarrow b_q=b_p$ {(for $q\neq p$)}, i.e., the structure is $\{a,a,c_1\dots\}$ with $a\neq0$, in agreement with~(\ref{A12_deg_Phi3}). For $\Phi_{z(p)v(p)}\neq0$, {the second equation of}~\eqref{r_deg_Phi3_constr} also implies $\{a,a,c_1\dots\}$ with $a\neq0$.

We are interested, in particular, in determining what are the necessary conditions in order to have all the non-vanishing eigenvalues distinct. From the {above} observations it follows that one necessarily has $\Phi_{(p)q(p)o}=\Phi_{(z)p(z)q}=\Phi_{z(p)v(p)}=0$, therefore the only non-zero $\Phi^{\{3\}}_{ijkl}$ components can be the $\Phi_{z(w)v(w)}$, however with the constraint $\Phi^0_{zwvw}=0$ (eq.~(\ref{r_deg_Phi3_constr})).  
It is easy to see that for this to be satisfied in a non-trivial way,  indices $z,v\ldots$  must run at least over four values, i.e. there must be at least four zero eigenvalues of $\rho_{ij}$ (excluding Kundt, the spacetime must thus be at least seven-dimensional). In other words:
\begin{itemize}
 \item all non-zero $\rho_{ij}$ are distinct and $m>n-6$ $\Rightarrow$ $\Phi^{\{3\}}_{ijkl}=0$. 
\end{itemize}

\subsubsection{$\Phi^{\{4\}}_{ijkl}$ components}
\label{subsubPhi4deg} 

From~\eqref{Bi3} and the symmetries of the Weyl tensor one finds
\beqn
 & & \Phi_{pqot}=\frac{\Phi_{pq(o)(t)}^0}{(r-b_{(o)})(r-b_{(t)})} , \qquad \Phi_{pzqw}=\frac{\Phi_{pz(q)w}^0}{r-b_{(q)}}  , \qquad \Phi_{zvwy}=\Phi_{zvwy}^0 , \\
 & & \Phi_{pqoz}=0 , \qquad \Phi_{pqzw}=0, \qquad \Phi_{pzvw}=0 .
 \label{r_deg_Phi4}
\eeqn
Similarly as in the case of $\Phi^{\{3\}}_{ijkl}$, we observe that 
{if} $\Phi_{pqot}\neq0$ then $b_t=b_q$ and $b_o=b_p$ (or $b_t=b_p$ and $b_o=b_q$), i.e., 
\be
 \Phi^{\{4\}}_{pqot}\neq0 \Rightarrow \rhob=\diag(a,a,b,b,c_1\dots) \qquad (a,b\neq0). 
 \label{Phi4_new}  
\ee

{Additionally,} {if} $\Phi_{pzqw}\neq0$ then 
$b_q=b_p$, {so that the structure is} { $\{a,a,c_1\dots\}$ with  $a\neq0$}.

Having  all the non-vanishing eigenvalues distinct thus requires $\Phi_{pqot}=0=\Phi_{pzqw}$ and the only non-zero components can be  $\Phi_{zvwy}$. Therefore we conclude again that  indices $z,v\ldots$  must run at least over four values unless $\Phi^{\{4\}}_{ijkl}=0$ (and, again, excluding Kundt, the spacetime must thus be at least seven-dimensional), i.e., 
\begin{itemize}
 \item all non-zero $\rho_{ij}$ are distinct and $m>n-6$ $\Rightarrow$ $\Phi^{\{4\}}_{ijkl}=0$. 
\end{itemize}

\subsection{Case $\Phi_{ij}=0$ (type II(abd))}

\label{subsec_Phi=0}

This case (non-trivial only for $n>5$) is obtained by setting
\be
 A=-\Phi^0_p=K^0 , \qquad \Phi^0_z=0 
\ee
{in the results obtained in~\ref{subsec_Phi_W_deg}.}
For the $W_{ij}$ components one thus has
\be
 W_{pq}=\frac{W^0_{pq}}{(r-b_{(p)})(r-b_{(q)})} , \qquad W_{pz}=\frac{W_{(p)z}^0}{r-b_{(p)}} , \qquad W_{zv}=W_{zv}^0 ,\label{Phi=0_Wpq}
\ee
with
\be
 \sum_vW_{zv}^0=0 , \qquad \sum_p\frac{W^0_{(p)z}}{r-b_{(p)}}=0 , \qquad \sum_zW^0_{qz}=0 , \qquad \sum_{p}\frac{W^0_{(p)q}}{r-b_{(p)}}=0 .
 \label{Phi=0_constr}
\ee

In view of these constraints (and of the properties of $W_{ij}$) we can briefly comment on some special cases:

\begin{itemize}
 \item $\Phi_{ij}=0$, $m=1$ $\Rightarrow$ $W_{pq}=0=W_{pz}$, 
 \item $\Phi_{ij}=0$, $m=2$ or $m=3$ $\Rightarrow$ $W_{pq}=0$, 
 \item $\Phi_{ij}=0$, $m=n-3$ $\Rightarrow$ $W_{zv}=0=W_{pz}$, 
 \item $\Phi_{ij}=0$, $m=n-4$ or $m=n-5$ $\Rightarrow$ $W_{zv}=0$. 
\end{itemize}

For special values of $n$ and $m$ some of these can hold simultaneously, thus leading to $W_{ij}=0$. Recalling the trivial implications $m=0\Rightarrow W_{pq}=0=W_{pz}$ and $m=n-2\Rightarrow W_{zv}=0=W_{pz}$ (valid also for $\Phi_{ij}\neq 0$), we have in particular 
\begin{itemize}
 \item for $n=4$ and $m=0,1,2$, $\Phi_{ij}=0$ $\Rightarrow$ $W_{ij}=0$, 
 \item for $n=5$ and $m=0,1,2,3$, $\Phi_{ij}=0$ $\Rightarrow$ $W_{ij}=0$, 
 \item for $n=6$ and $m=1,3$, $\Phi_{ij}=0$ $\Rightarrow$ $W_{ij}=0$. 
\end{itemize}
While the first two implications simply reproduce the known result that $\Phi_{ij}=0\Leftrightarrow W_{ij}=0$ for $n=4,5$ (for any permitted $m$), the last remark will be useful for later purposes.

Using \eqref{Phi=0_Wpq}, \eqref{Phi=0_constr} and a reasoning similar to that of section \ref{subsec_onerep} (in the paragraph after eq.~\eqref{Wij}, including footnote~\ref{foot_311}) one can also show that
\be
 W_{pq}\neq0 \Rightarrow \rhob=\diag(a,a,b,b,c_1\dots) \qquad (a,b\neq0),  
 \label{W_new}
\ee
which will be useful in the following.

\subsubsection{No repeated non-zero eigenvalues}

\label{subsubsec_Phi=0_norepeated}

In case all the $b_p$ are distinct, {equations~(\ref{Phi=0_constr})} imply
\be
 W^0_{pz}=0=W^0_{pq} ,
\ee
so that the only non-zero components of $W_{ij}$ are the $W_{zv}=W_{zv}^0$, with $\sum_vW_{zv}^0=0$. This implies that, in order to have $W_{ij}\neq0$, one needs that  indices $z,v,\ldots$ run at least over  four values, i.e., the cases $m=n-3$, $m=n-4$ and $m=n-5$ imply $W_{ij}=0$. Thus we have proven: $\Phi_{ij}=0$, all non-zero $\rho_{ij}$ are distinct and $m>n-6$ $\Rightarrow$ $W_{ij}=0$. {Together with the results of subsection~\ref{subsec_Phi3_4_deg}, we see that if such assumptions hold all boost weight zero components must actually vanish, so that we can conclude}

\begin{prop} 
\label{prop_IIabd}
For type II(abd) ($\Phi_{ij}=0$) non-twisting Einstein spacetimes with (degenerate) $\rhob$ of rank $n-6<m<n-2$, the eigenvalue structure of $\rhob$ is 
$\{a,a, c_1\ldots, 0,\ldots\}$ (with $a\neq0$) or more special {(i.e., at least two non-zero eigenvalues of $\rhob$ coincide)}.
\label{prop_Phi=0}
\end{prop}

The condition on $m$ means that either $m=n-3$ or $m=n-4$ or $m=n-5$. We will see {in section~\ref{subsec_n-3}} below that in the case $m=n-3$ the assumption $\Phi_{ij}=0$ can in fact be dropped {(since $\Phi_{ij}\neq0$ is not possible if all non-zero eigenvalues are distinct).  For $m=n-4$ the eigenvalue structure is thus $\{a,a,0,0,c_1,\ldots,c_{n-6}\}$ and for $m=n-5$ it is $\{a,a,0,0,0,c_1,\ldots,c_{n-7}\}$ (while for any smaller $m$ there are, of course, $n-2-m\ge 4$ vanishing eigenvalues).} 
Recall also that for $m=n-2$ (non-degenerate case) we had a similar {result, cf.~proposition~\ref{prop_non_deg2}.}

\subsection{The special case $m=n-3$}

\label{subsec_n-3}

For $n=4$ this case is not possible since it gives $m=1$ (see section~\ref{subsec_m=1}). 
{In higher dimensions, the simplest spacetimes of this class are obtained by taking a Brinkmann warp {\cite{Brinkmann25} (see also \cite{OrtPraPra11})} of a type II Einstein spacetime possessing a non-degenerate mWAND with a single (spacelike) extra-dimension (cf. also \cite{OrtPraPra11}). They include, e.g., {static} black strings and} this case may thus be of special interest. 

Since there is only one vanishing eigenvalue of $\rho_{ij}$, i.e., $z$ can take only one value, we have $W_{zv}=0$ and, from~(\ref{constr_Phi_z}),
\be
 \Phi^0_z=0 , \qquad \sum_p\frac{W^0_{pz}}{r-b_p}=0 .
 \label{constr_m=n-3}
\ee

Therefore (\ref{constrW_deg_leading}) reduces to
\be
 (n-3)\Phi^0_{q}+W^0_{qz}=\sum_p\Phi^0_{p} .
\ee

From the results of subsection~\ref{subsec_Phi3_4_deg}, when $m=n-3$ the only possible non-zero components of $\Phi^{\{3\}}_{ijkl}$ can be $\Phi_{(p)q(p)o}$, with $q\neq o$ (in particular, $\Phi_{pzqz}=\Phi_{p(z)q(z)}=0$,  see the first of \eqref{r_deg_Phi3_constr_0}), which satisfy
\be
 \sum_o\frac{\Phi^0_{p(o)q(o)}}{r-b_{(o)}}=0 \qquad (p\neq q) , 
\ee
while the only possible non-zero components of $\Phi^{\{4\}}_{ijkl}$ are $\Phi_{pqot}$ ($\Phi_{pqoz}=0$, see the first of \eqref{r_deg_Phi4}) . 
From sections \ref{subsubPhi3deg} and \ref{subsubPhi4deg} (see \eqref{Phi3_new}, \eqref{Phi4_new})     
it thus follows that if $\Phi^{\{3\}}_{ijkl}\neq0$ ($\Leftrightarrow\Phi^{\{3\}}_{(p)q(p)o}\neq0$ here) or $\Phi^{\{4\}}_{ijkl}\neq0$ ($\Leftrightarrow\Phi^{\{4\}}_{pqot}\neq0$ here)  then the eigenvalue structure is  $\{a,a,b,b,c_1\dots\}$, where  $a,b\neq0$. Moreover, it is easy to see that when $m=n-3$ both $\Phi^{\{3\}}_{ijkl}$ and $\Phi^{\{4\}}_{ijkl}$  can be non-zero only for $n\ge7$.  This is of interest only in six dimensions ($\Phi^{\{3\}}_{ijkl}=0=\Phi^{\{4\}}_{ijkl}$ identically for $n=4,5$) and can be summarized as 
\begin{itemize}
 \item $n=6$, $m=3$  $\Rightarrow$ $\Phi^{\{3\}}_{ijkl}=0=\Phi^{\{4\}}_{ijkl}$. 
\end{itemize}

If we additionally assume $\Phi_{ij}=0$ (then also $W_{zv}=0=W_{pz}$, see comments after \eqref{Phi=0_constr}), then some of the components $W_{pq}$, $\Phi^{\{3\}}_{p(o)q(o)}$ or $\Phi^{\{4\}}_{pqot}$ have to be non-vanishing (otherwise the spacetime would be of type III). Then using also results of section~\ref{subsec_Phi=0}  {(eq.~\eqref{W_new})} it follows that type II spacetimes {with $m=n-3$ and $\Phi_{ij}=0$} {must have} the eigenvalue structure $\{a,a,b,b,c_1\dots c_{n-7},0\}$.

\subsubsection{No repeated non-zero eigenvalues}

\label{subsubsec_m=n-3_no_rep}

If all $b_p$ are distinct from the second of~(\ref{constr_m=n-3}) we obtain  
\be
 W^0_{pz}=0 ,
\ee
so that $W_{pz}=0$ and all the $\Phi^0_{q}$ coincide, namely, 
\be
 \Phi^0_{q}=\frac{1}{(n-3)}\sum_p\Phi^0_{p}\equiv f^0 .
\ee
Therefore
\be
 A=-f^0+A^0\prod_p\frac{1}{r-b_p} ,
\ee
and one can basically proceed as in the non-degenerate case. Comparing the structures of the poles of the l.h.s. and r.h.s. of (\ref{constrW_deg}) one again arrives at
\be
 A^0=0, \qquad W^0_{pq}=0 .
\ee
Combining this with the previous results we conclude $\Phi_q=0$, $\Phi_z=0$, $W_{qz}=0$, $W_{qp}=0$, $W_{zv}=0$, i.e.
\be
 \Phi_{ij}=0=W_{ij} .
\ee
Since here $m=n-3>n-6$, we can use the results of subsection~\ref{subsec_Phi3_4_deg} to conclude that also $\Phi^{\{3\}}_{ijkl}=0=\Phi^{\{4\}}_{ijkl}$, so that all b.w. 0 components of the Weyl tensor vanish. Thus we have proven: {\it all non-zero $\rho_{ij}$ are distinct and $m=n-3$ $\Rightarrow$ all b.w. 0 components vanish.}
In other words, {\it in the case $m=n-3$, the eigenvalue structure of $\rhob$ must be $\{a,a,c_1, \dots ,c_{n-5}, 0\}$ 
(or more special)}. If, additionally, $\Phi_{ij}=0$, there is in fact another repeated non-zero eigenvalue \cite{Wylleman_priv}\footnote{We thank Lode Wylleman for pointing this out.} {(see the comment just before section \ref{subsubsec_m=n-3_no_rep}), so that} 
 Proposition~\ref{prop_Phi=0} can accordingly be reformulated as (see also \cite{Wylleman_priv}) 
 
\begin{prop} 
\label{prop_n-3_text}
In a type II Einstein spacetime of dimension $n>4$ with a non-twisting multiple WAND of rank $m=n-3$, the eigenvalue structure of $\rhob$ is $\{a,a,c_1,\ldots c_{n-5},0\}$ where $a,c_\alpha  \ne 0$ (and necessarily $\Phi^A_{ij}=0$ if $n>5$). If, additionally, $\Phi_{ij}=0$, the structure becomes $\{a,a,b,b,c_1,\ldots c_{n-7},0\}$ where $a,b,c_\alpha  \ne 0$.
\end{prop}

\begin{rem}
\label{rem_Phi0}
{Together with propositions~\ref{prop_non_deg2} and \ref{prop_Phi=0}, this shows that, for any possible value of $m$ ($0<m\le n-2$), the mWAND of type II(abd) (i.e., $\Phi_{ij}=0$) Einstein spacetimes has  $\rhob$ of the form $\{a,a,b,b,c_1,\ldots c_{n-6}\}$ (where $a\ne 0$ if $m=n-5$ or $m=n-4$ and $a,b\ne 0$ if $m=n-3$ or $m=n-2$ ).} In particular, point~(\ref{Phi0}) of theorem~\ref{prop_GSHD} is thus now also proven.
\end{rem}

\begin{rem}
\label{rem_5d}
In five dimensions (for which always $\Phi_{ij}\neq0$), recalling also the results for the non-degenerate case {(proposition~\ref{prop_non_deg2})}, this enables us to list all the admitted eigenvalues structures, {for $m=3,2,1,0$, respectively}: $\rho_{ij}=\mbox{diag}(a,a,a)$, 
$\rho_{ij}=\mbox{diag}(a,a,0)$,  
$\rho_{ij}=\mbox{diag}(a,0,0)$,  
or $\rho_{ij}=0$ (Kundt). See \cite{Ortaggioetal12} for more details and examples. 
\end{rem}

\begin{rem}
\label{rem_6d}
In six dimensions, we conclude that the case $m=3$ requires $\rho_{ij}=\mbox{diag}(a,a,b,0)$. However, as shown above in this case one has $\Phi^{\{3\}}_{ijkl}=0=\Phi^{\{4\}}_{ijkl}$ and additionally the implication $\Phi_{ij}=0$ $\Rightarrow$ $W_{ij}=0$ holds (see subsection~\ref{subsec_Phi=0}). Therefore one necessarily has $\Phi_{ij}\neq0$ when $n=6$ and $m=3$ (or all Weyl boost zero components would vanish). {See section~\ref{sec_6D} for more details.}
\end{rem}

\subsection{The special case $m=1$} 

\label{subsec_m=1}

From a viewpoint complementary to that of the previous subsection, it is also interesting to analyze the case when only one eigenvalue is non-zero. Let us associate the index $p=2$ to this eigenvalue, while $z,v=3,\ldots,n-1$. Using the general equations presented above, one obtains (cf. also~\cite{PraPra08}) 
\be
 \Phi_z=\Phi^0_z , \qquad \Phi_2=-\frac{1}{2}\sum_z\Phi^0_z, \qquad W_{2z}=\Phi^0_z , \qquad \sum_vW^0_{zv}=-3\Phi^0_z , 
\ee
while clearly here $W_{pq}=0$ (note that for $n=4$ one gets 
$\Phi^0_z=0$  
and $\Phi_{ij}=0=W_{ij}$, so this case does not occur). Therefore all components of $\Phi_{ij}$ and $W_{ij}$ are $r$-independent. Note that  $\Phi_{ij}=0\Leftrightarrow \Phi^0_z=0$. When this happens the only non-zero components of $W_{ij}$ can be the $W^0_{zv}$, but with the constraint $\sum_vW^0_{zv}=0$. Therefore it is easy to see that $W_{ij}$ can be non-zero only if $z,v$ can take at least four values, i.e., $z,v=3,4,5,6,\ldots$, so that $n\ge 7$. To summarize \begin{itemize}
 \item $n<7$, $m=1$, $\Phi_{ij}=0$ $\Rightarrow$ $W_{ij}=0$ .
\end{itemize}
This can be understood as a special instance of the results of subsection~\ref{subsec_Phi=0}. 
This result is of interest in six dimensions, since one always has $\Phi_{ij}=0\Rightarrow W_{ij}=0$ in four and five dimensions. Namely, for $n=6$ the case $\{a ,0,0,0\}$ is forbidden if $\Phi_{ij}=0\neq W_{ij}$. 

From the results of subsection~\ref{subsec_Phi3_4_deg}, when $m=1$ the only possible non-zero components of $\Phi^{\{3\}}_{ijkl}$ can be $\Phi_{(z)v(z)w}$ (with $v\neq w$), which satisfy
\be
 \Phi^0_{zvzw}=0 \qquad (v\neq w) , 
\ee
while the only possible non-zero components of $\Phi^{\{4\}}_{ijkl}$ are $\Phi_{zvwy}$. It is easy to see that, since $m=1$,  both sets of components vanish identically unless $n\ge7$, so that 
\begin{itemize}
 \item $n=6$, $m=1$  $\Rightarrow$ $\Phi^{\{3\}}_{ijkl}=0=\Phi^{\{4\}}_{ijkl}$. 
\end{itemize}

\begin{rem}
\label{rem_6d_m=1}
Since for $n=6$, $m=1$ one has $\Phi_{ij}=0$ $\Rightarrow$ $W_{ij}=0$ (see above), we conclude that in six dimensions spacetimes with $m=1$ require $\Phi_{ij}\neq 0$.
\end{rem}

\subsubsection{Examples}

\label{subsubsec_m=1}

For any $n\ge5$, simple examples of spacetimes with $m=1$ {(and thus $\Phi^A_{ij}=0$)} are given by:
\begin{enumerate}
 \item dS$_3\times$S$_{n-3}$ (or AdS$_3\times$H$_{n-3}$): $n\ge5$, $\Phi_{ij}\neq 0$, $\Lambda\neq 0$, \label{dSxS_m=1}
 \item as above but with a Brinkmann warp: $n\ge6$, $\Phi_{ij}\neq 0$, also $\Lambda=0$ is possible, \label{Brink_m=1} 
 \item Minkowski$_3\times$(Ricci-flat)$_{n-3}$: $n\ge7$, $\Phi_{ij}=0$, $\Lambda=0$. \label{MinkxRc_m=1}
\end{enumerate}

In the first case, if {we restrict to} $n\ge7$ the second factor space can be in fact {\em any} Einstein space (with Ricci scalar given by $2(n-3)K$). The metric of the first and third case can thus be given in a unified way as
\be
 \d s^2 =\left(1+\frac{K}2ur\right)^{-2}(2\d u\d r+r^2\d x^2)+\d\Sigma^2 ,
\ee
where $\d\Sigma^2$ is the metric of an $(n-3)$-dimensional Einstein space with Ricci scalar $R_\Sigma=2(n-3)K$ and $K$ is a constant related to the $n$-dimensional Ricci scalar by $R=2nK$. A geodesic twistfree mWAND $\bl$ (degenerate but expanding along $x$) is
\be
 {\ell}_a\d x^a=\d u .
\ee
Note that here $r$ is an affine parameter along $\bl$ only in the case $K=0$. {Such product spaces are of type D \cite{PraPraOrt07} and in fact any null vector field tangent to the three-dimensional Lorentzian factor is an mWAND \cite{GodRea09,PraPraOrt07}.}
Using these results metrics for the second case can also be constructed straightforwardly. {These will be also of type D \cite{OrtPraPra11}.}

\subsection{Case when all non-zero eigenvalues are distinct} 

\label{subsec_alldistinct}

By combining the results of subsections~\ref{subsec_Phi=0} and \ref{subsec_Phi3_4_deg} we have 
\begin{itemize}
 \item all non-zero $\rho_{ij}$ are distinct, $m=n-4$ or $m=n-5$ and $\Phi_{ij}=0$ $\Rightarrow$ all b.w. 0 components vanish,
 \item all non-zero $\rho_{ij}$ are distinct and $m=n-3$ $\Rightarrow$ all b.w. 0 components vanish.\footnote{By adding the results of \cite{Pravdaetal04} we can in fact conclude that in both cases the whole Weyl tensor must vanish.} 
\end{itemize}
Therefore type II Einstein spacetimes for which all non-zero $\rho_{ij}$ are distinct require that at least one of the following holds ({obviously $m>0$ or we would simply have Kundt}):
\begin{enumerate}
 \item $1\le m\le n-6$ (i.e., at least four eigenvalues of $\rho_{ij}$ vanish),
 \item $\Phi_{ij}\neq 0$ and $1\le m\le n-4$ (i.e., at least two eigenvalues of $\rho_{ij}$ vanish).
\end{enumerate}
The first case can occur only for $n\ge7$: for $n=7$ one can have only $\rho_{ij}=\mbox{diag}(a ,0,0,0,0)$, for $n=8$ either $\rho_{ij}=\mbox{diag}(a ,0,0,0,0,0)$ or  $\rho_{ij}=\mbox{diag}(a ,b ,0,0,0,0)$, etc.. The second case can occur for $n\ge5$: for $n=5$ one has simply $\rho_{ij}=\mbox{diag}(a ,0,0)$, for $n=6$ either $\rho_{ij}=\mbox{diag}(a ,0,0,0)$ or $\rho_{ij}=\mbox{diag}(a ,b ,0,0)$, for $n=7$ either $\rho_{ij}=\mbox{diag}(a ,0,0,0,0)$ or $\rho_{ij}=\mbox{diag}(a ,b ,0,0,0)$ or $\rho_{ij}=\mbox{diag}(a ,b ,c ,0,0,0)$, etc.. These results are summarized in table~\ref{tab_distinct_eigenval}. For such solutions {\em the optical constraint is clearly violated if $m>1$}. The {previously discussed} case $m=1$ has been included for completeness but it is of course ``trivial'' in this context since there is a single non-zero eigenvalue.

\begin{table}[t]
 \begin{center}
  \begin{tabular}{|c|c|l|c|}
   \hline Case & $n$ & Possible $m$ & Examples  \\\hline
    $\Phi_{ij}\not=0$ & 5 & 1 & \\
    					   		  & 6 & 1,2 &   \\
     									& 7 & 1,2,3 &  \\
     							 		& 8 & 1,2,3,4 &  \\
     							 		& \vdots &  &  \\
     								 	& $n$ & 1,2,3\ldots, $n-4$ & (A)dS$_{m+2}\times$S$_{n-2-m}$(H$_{n-2-m}$) \\\hline
     $\Phi_{ij} =0$   & 7 & 1 &  \\
        							& 8 & 1,2 &  \\
        						  & \vdots &  &  \\
        							& $n$ & 1,2,3,\ldots, $n-6$ &  Mink$_{m+2}\times$(Ricci-flat)$_{n-2-m}$  \\
\hline
  \end{tabular}
    \caption{Non-twisting Einstein spacetimes for which there are no repeated non-zero eigenvalues of $\rho_{ij}$: permitted values of the rank $m$ of $\rho_{ij}$ in various dimensions {(necessarily $1\le m \le n-4$ thanks to propositions~\ref{prop_non_deg2}  and \ref{prop_n-3_text})}. {As discussed in the text, different examples can also be generated using Brinkmann's warp, or by replacing the S$_{n-2-m}$(H$_{n-2-m}$) factor by a more general Einstein space. {Note that here solutions with $m>1$ violate the optical constraint.\label{tab_distinct_eigenval}}
    }}
  \end{center}
\end{table}

\subsubsection{Examples}

\label{subsubsec_examples_distinct}

The case $n=4$ is not possible. Examples with $m=1$ have been already presented {in subsection~\ref{subsec_m=1} and will not be discussed again here}. For $m\ge2$ the $n=5$ case is also forbidden {(cf. Remark~\ref{rem_5d} and \cite{Ortaggioetal12})}. 
For $n=6$ the only possibility is $m=2$ (with $\Phi_{ij}\neq0$) and one can construct explicit Einstein spacetimes by taking dS$_4\times$S$^2$ or AdS$_4\times$H$^2$. These can then be extended to any dimension $n\ge7$ by a simple Brinkmann warp (and if one starts from dS$_4\times$S$^2$ the cosmological constant of the resulting spacetime can be arbitrary, cf.~\cite{OrtPraPra11}), so to have explicit examples of the type $\rho_{ij}=\mbox{diag}(a ,b ,0,\ldots,0)$ (i.e., with $m=2$) for any $n>5$. 

However, according to the previous comments (see also table~\ref{tab_distinct_eigenval}), for $n=7$ one can have $m=2,3$, for $n=8$ it is possible $m=2,3,4$ and so on: in general, and for any $n>5$, $2\le m\le n-4$ explicit metrics can be constructed similarly as in six dimensions by taking dS$_{m+2}\times$S$_{n-2-m}$ or AdS$_{m+2}\times$H$_{n-2-m}$ (where the two factor spaces must have {Ricci scalars} given by $(m+2)(m+1)K$ and $(n-2-m)(m+1)K$, respectively), and appropriate hypersurface orthogonal null congruences living in the (A)dS$_{m+2}$ factor {(explicit examples are given below)}. These have necessarily $\Phi_{ij}\neq0$ and $\Lambda\neq0$. Using Brinkmann's warp one can also generate Ricci-flat solutions (again with $\Phi_{ij}\neq0$) for which, however, the stronger restriction $n>6$, $2\le m<n-4$ holds (at least if this method is used). Products of Ricci-flat spaces are possible if each factor is at least four-dimensional (therefore there are always at least four vanishing eigenvalues) and the Lorentzian factor is flat and fall into the $\Phi_{ij}=0$ class, with $\Lambda=0$ (these metrics belong to the type D \pp waves mentioned in \cite{OrtPraPra09}). 

Summary:
\begin{enumerate}
 \item dS$_{m+2}\times$S$_{n-2-m}$ (or AdS$_{m+2}\times$H$_{n-2-m}$)\footnote{Note that if $n-2-m\ge4$ the second factor space can be in fact any Einstein space with Ricci scalar given by $(n-2-m)(m+1)K$.}: $n>5$, $2\le m\le n-4$, $\Phi_{ij}\neq 0$, $\Lambda\neq 0$, \label{dSxS}
 \item as above + Brinkmann, also $\Lambda=0$: $n>6$, $2\le m\le n-5$, $\Phi_{ij}\neq 0$, \label{Brink}
 \item Minkowski$_{m+2}\times$(Ricci-flat)$_{n-2-m}$: $n>7$, $2\le m\le n-6$, $\Phi_{ij}=0$, $\Lambda=0$. \label{MinkxRc}
\end{enumerate}

Note that all the above spacetimes {also} belong to the Kundt class, although w.r.t. a null congruence different from the one which is of interest for our discussion (which is expanding and shearing). {In fact they all admit $\infty^m$ mWANDs:} similarly as in section~\ref{subsubsec_m=1}, they are of type D (with $\Phi^A_{ij}=0$) and any null vector field tangent to the Lorentzian factor is an mWAND -- so one can always also find an mWAND with an optical matrix having all non-zero eigenvalues identical.  {However, this mWAND ``degeneracy'' is not the generic situation} (see appendix~\ref{app_violating} for an example of a class of type D spacetimes which admit (only) two double WANDs, both violating the optical constraint).

Explicitly, for cases~\ref{dSxS} and \ref{MinkxRc} one can take, e.g., the metric
\be
 \d s^2 =\Omega^2(2\d u\d r+P_q^2\d x_q^2)+\d\Sigma^2 ,
 \label{ex_distinct}
\ee
where $q=2,\dots,m+1$, $\d\Sigma^2$ is the metric of an $(n-2-m)$-dimensional Einstein space with Ricci scalar $R_\Sigma=(n-2-m)(m+1)K$ and
\be
 P_q=r-b_q, \qquad \Omega^{-1}=1+\frac{K}{4}\left[2r\left(u-\frac{1}{2}P_qx_q^2\right)+P_q^2x_q^2\right ] ,
\ee
{in which $b_q$ and $K$ are constants.}
The first factor space in~(\ref{ex_distinct}) is a space of constant curvature with Ricci scalar given by $(m+2)(m+1)K$, so that the Ricci scalar of the full spacetime is 
\be
 R=n(m+1)K .
\ee

The geodesic, twistfree mWAND $\bl$ and the remaining frame vectors are given by
\be
 {\ell}_a\d x^a=\d u , \qquad n_a\d x^a=\Omega^2\d r, \qquad m_{(q)a}\d x^a=\Omega P_{(q)}\d x_{(q)} ,
\ee
while vectors $m_{(v)a}$ will depend on the specific form of $\d\Sigma^2$. The optical matrix $\rho_{ij}$ is given by
\be
 \rho_{pq}=\delta_{(p)q}\frac{(\Omega P_{(p)})_{,r}}{\Omega^{3}P_{(p)}} , \qquad \rho_{pv}=0=\rho_{vz} .
\ee
In the case $K=0$ (i.e., $\Lambda=0$) one thus has $\rho_{pq}=\delta_{(p)q}/(r-b_{(p)})$. Note, however, that $r$ is {\em not} an affine parameter when $K\neq 0$ (one has $\bl=\Omega^{-2}\pa_r$).

For case~\ref{Brink} one can use the method illustrated in detail in \cite{OrtPraPra11}.

\section{Counterexamples}

\label{sec_counter}

In \cite{Ortaggioetal12}, a counterexample to the converse of the five-dimensional ``shear-free'' part of the GS theorem was presented, thus demonstrating explicitly that the condition that the optical matrix $\rhob$ admits a canonical form compatible with a geodetic mWAND is not sufficient for the null geodetic being  an mWAND. Similarly, also our results above give conditions that are (necessary but) not sufficient. In order to demonstrate that, here we present a few ``counterexamples'', i.e., certain Einstein spacetimes that admit a non-twisting null (thus geodesic) vector field $\bl$ with $\rhob$ taking one of the permitted ``canonical forms'' (cf. theorem~\ref{prop_GSHD} and propositions~\ref{prop_non_deg2}, \ref{prop_IIabd} and \ref{prop_n-3_text}) and yet with $\bl$ not being an mWAND. Note that such counterexamples will necessarily be shearing: a twistfree shearfree null vector field is automatically an mWAND \cite{PodOrt06,OrtPraPra07}.

\subsection{Non-degenerate $\rhob=\{a,a,b,b,c_1, \ldots, c_{n-6} \}$}

In the non-degenerate case a ``counterexample'' to the canonical form $\{a,a,b,b,c_1, \ldots, c_{n-6} \}$ (with $a,b,c_\alpha \ne 0$ and $a\neq b$ -- see points (ii) and (iii) of theorem~\ref{prop_GSHD} and proposition~\ref{prop_non_deg2}) is given by the following six-dimensional Ricci-flat spacetime, which belongs to a class of metrics considered by Robinson (as described in \cite{Trautman02b}),
\be
 \d s^2=2\d u\d r+2\cosh^2r\d w\d\bar w+2\sin^2r\d\zeta\d\bar\zeta .
 \label{robinson}
\ee
Here the hypersuface orthogonal null vector field $\ell_a\d x^a=\d u$ is {\em not} an mWAND (not even a single one) and yet the corresponding 
$\rhob$ has the eigenvalue structure $\{a,a,b,b\}$ (with $a\neq0\neq b$, $a\neq b$). Associated to $\bl$ there is also an optical structure \cite{Trautman02b}.
Note, however, that spacetime~\eqref{robinson} is a type N \pp wave (albeit in {so-called Rosen coordinates, see, e.g., metric~(39) of \cite{Brinkmann25}}), with a covariantly constant mWAND given by $n_a\d x^a=\d r$. 

\subsection{Case $n-6<m<n-2$} 

Direct products of~\eqref{robinson} with flat extra dimensions clearly provide non-WANDs with eigenvalues $\{a,a,b,b,0,\ldots,0 \}$. In particular, for $n=7$ one obtains $\{a,a,b,b,0\}$, which is thus a counterexample to the $\rhob$-form of Proposition~\ref{prop_n-3_text}. For $n=8,9$ this is also a counterexample to the $\rhob$-form of Proposition~\ref{prop_IIabd}. Another counterexample to Proposition~\ref{prop_n-3_text} can be obtained for $n=6$ by taking the direct product of a vacuum black ring \cite{EmpRea02prl} with a flat dimension: in the coordinates of \cite{ElvEmp03}, the hypersurface orthogonal null vector field  (also considered in 5d in \cite{Taghavi-Chabert11} for different purposes)
\be
 \ell_a\d x^a = \frac{\sqrt{F(y)}}{G(y)}\d y+\d\psi ,
\ee
corresponds to $\rhob$ of the form $\{a,a,b,0\}$ and is not a WAND. 
Alternatively, the same form $\{a,a,b,0\}$ can be obtained by taking the direct product of a 5d vacuum static KK bubble \cite{GodRea09} with a flat dimension (with a non-WAND $\ell_a\d x^a=\d t+V(r)^{-1/2}\d r$ -- this 6d spacetime is of type D but the mWANDs are different from $\bl$ \cite{OrtPraPra11}).

\subsection{Case when all non-zero eigenvalues are distinct} 

Examples of spacetimes with a non-twisting, non-WAND $\bl$ having all non-zero eigenvalues distinct can be constructed as a direct product of a 4d \pp wave with flat space, i.e.,
\be
 \d s^2=2\d u\d r+\cosh^2r\d x^2+\sin^2r\d y^2+\d z_i\d z^i .
 \label{4d_pp}
\ee
Here,  $\ell_a\d x^a=\d u$ has a $\rhob$ with eigenvalue structure $\{a,b,0,\ldots,0\}$ (with $a\neq0\neq b$, $a\neq b$). This is the form of $\rhob$ discussed in section~\ref{subsec_alldistinct} (see also table~\ref{tab_6D} for the $n=6$ case). Similarly as the spacetime~\eqref{robinson}, also \eqref{4d_pp} is a type N \pp wave with mWAND $n_a\d x^a=\d r$. 

\subsection{Case $m=1$}

Examples of non-twisting null congruences with $m=1$ (the case discussed in section~\ref{subsec_m=1}) that are not multiple WANDs can be obtained from the (Ricci-flat) Newman-Tamburino lift considered in \cite{Ortaggioetal12}, i.e.,
\be
{\rm d}s^2 = r^2 {\rm d} x^2 + x^2 {\rm d} y^2 - \frac{4 r}{x } {\rm d} u {\rm d} x - 2 {\rm d} u {\rm d} r + x^{-2} \left[c + \ln (r^2 x^4)\right] {\rm d} u^2+\d z_i\d z^i , \label{NTmetric}
\ee
where $c$ is a constant. In this case $\ell_a\d x^a = {\rm d} u$ is a {\em single} WAND  with optical matrix of the form $\{a,0,\ldots,0\}$ (with $a\neq 0$).

\section{Existence of totally geodesic null two-surfaces}

\label{sec_totgeod}

Theorem~\ref{prop_GSHD} comprises the main results proven in the previous sections. In particular, it implies that there always exists at least a pair of repeated (possibly vanishing) eigenvalues and that (recall that $n\ge 6$) there are in fact at least two such pairs when, e.g.,  $\det\rhob\neq0$ (see also proposition~\ref{prop_non_deg2}) or $\Phi_{ij}=0$ (see also propositions~\ref{prop_IIabd} and \ref{prop_n-3_text}). The following proposition (see also the definitions in appendix~\ref{app_shearfree}, in particular eq.~\eqref{D}) elucidates the geometrical meaning of this fact (and provides a connection with the standard 4d version of the Goldberg-Sachs theorem, cf. appendix~\ref{app_GS_4D}).
\begin{prop} 
 In an algebraically special Einstein spacetime of dimension $n\ge 6$ with a non-twisting {mWAND} $\bl$, using a parallelly transported eigenframe of $\rhob$
 \begin{enumerate} 
 	\item if the eigenvalue structure of $\rhob$ is $\{a,a,c_1,\ldots, c_{n-5},0 \}$ with $c_\alpha\neq a\neq 0$ then the totally null distribution ${\cal D}_{23}$ is integrable with totally geodesic integral surfaces, 
 	\item if the eigenvalue structure of $\rhob$ is $\{0,0,c_1,\ldots, c_{n-4}\}$ with $c_\alpha\neq 0$ then the totally null distribution ${\cal D}_{23}$ is integrable with totally geodesic integral surfaces,
 	\item if the eigenvalue structure of $\rhob$ is $\{a,a,b,b,c_1, \ldots, c_{n-6} \}$ with $b\neq a$ and $a\neq c_\alpha\neq b$ then the totally null distributions ${\cal D}_{23}$ and ${\cal D}_{45}$ are integrable with totally geodesic integral surfaces (either $a$ or $b$ can be zero). 
 \end{enumerate}	
\label{prop_integrab}
\end{prop}
\begin{proof}
The proof is similar in all cases and relies on the use of the Ricci identity (11k,\cite{OrtPraPra07}) (equivalent to the NP equation (A4,\cite{Durkeeetal10})), which for algebraically special Einstein spacetimes reduces to 
\be
 \delta_{[j|} \rho_{i|k]} =  L_{1[j|} \rho_{i|k]}+ \tau_i \rho_{[jk]}+\rho_{il}\M{l}{[j}{k]}
+ \rho_{l[j|}\M{l}{i|}{k]} , \label{11k}
\ee
where $\delta_i=m_{(i)}^a\nabla_a$ (recall also the definitions in \eqref{L1i_M}). Since here $\bl$ is twistfree we have $\rho_{[jk]}=0$ in the above equation. Assume the form of $\rhob$ specified in 1., i.e., $\rho_{22}=\rho_{33}=a\neq 0$ (and all remaining eigenvalues different from $a$) and consider \eqref{11k} for $i,j,k=23\hat k/32\hat k$ (with $\hat k\neq 2,3$). This gives $\M{2}{\hat k}{3}=0=\M{3}{\hat k}{2}$. Then, \eqref{11k} with $i,j,k=22\hat k/33\hat k$ gives $\M{2}{\hat k}{2}=\M{3}{\hat k}{3}$,  which completes the proof of point 1. (see eqs.~\eqref{2totgeod_real}). The proof of point 2. is identical. The proof of point 3. is analogous, after extending the reasoning to the index pair (45); in this case, however, either $a$ or $b$ can be zero, provided all $c_\alpha\neq 0$. 
\end{proof}
\begin{rem}
In the case there are further pairs of repeated eigenvalues, the corresponding two-spaces will also define integrable totally geodesic distributions, provided these pairs are not repeated. In particular in even dimensions, if there are $(n-2)/2$ (non-repeated) pairs of equal eigenvalues there will be $(n-2)/2$ such totally geodesic two-surfaces.
\end{rem}
\begin{rem}
Point 1. of the above proposition includes, in particular, all non-twisting but expanding type III/N Einstein spacetimes (for which $c_{\alpha}=0$ \cite{Pravdaetal04}) and type II non-twisting non-Kundt spacetimes with $\Phi^A_{ij} \ne 0$ (point \eqref{PhiA} of theorem~\ref{prop_GSHD}). 
Point 3. is relevant to the case $\det\rhob\neq0$ with $\Phi_{ij}=0$ (proposition~\ref{prop_non_deg2} and point \eqref{Phi0} of theorem~\ref{prop_GSHD}), which can only be of type II(abd).
\end{rem}
\begin{rem}
The proposition does not explicitly include the $n=5$ case. However, from Proposition~4 of \cite{Ortaggioetal12} one can easily obtain the corresponding results (also in the presence of twist). In particular, in the non-twisting case ${\cal D}_{23}$ is totally geodesic when the eigenvalue structure of $\rhob$ is either $\{a,a,0\}$ or  $\{0,0,a\}$ (with $a\neq 0$ in both cases). The same is true in the case $\{a,a,a\}\neq0$, although this does not follow from (11k,\cite{OrtPraPra07}) (instead, one can easily adapt the argument of footnote~10 of \cite{Ortaggioetal12}). 
\end{rem}
\begin{rem}
Recall that in 4d the Goldberg-Sachs theorem implies that $\bl$ is geodesic and the eigenvalue structure of $\rho_{(ij)}$ is $\{a,a\}$, which is {\em equivalent} to ${\cal D}_{23}$ being integrable (and thus automatically totally geodesic, in 4d, see appendix~\ref{app_GS_4D}).
\end{rem}

\section{Type II spacetimes with a non-twisting multiple WAND in six dimensions}

\label{sec_6D}

As already mentioned, five-dimensional algebraically special Einstein spacetimes have been studied in detail in \cite{Ortaggioetal12} and, for $n=5$, the results obtained above are also contained in that reference. The next lower dimension to consider is thus $n=6$ and in this section our general results are specialized to this case. This is of interest also because of a qualitative difference between $n<6$ and $n\ge 6$ dimensions: for $n<6$ type II spacetime necessarily have $\Phi_{ij}\neq 0$ (since type II coincides with type II(c) {for $n=4,5$}), while this is not so in higher dimensions.  

\subsection{Permitted forms of $\rhob$}

\label{subsec_permitted_6D}

One has to consider the various cases $m=4,3,2,1$ ($m=0$ is Kundt).

First, recall that for $m\neq 2$ point \eqref{PhiA} of theorem~\ref{prop_GSHD} implies that $\Phi^A_{ij}=0$. Let us thus first discuss this case.

Now, for $m=4$ proposition~\ref{prop_non_deg2} tells us that if $\Phi_{ij}\neq 0$ we are in the \RT class, i.e. $\rhob=\mbox{diag}(a,a,a,a)\neq 0$ (in which case $\Phi_{ij}\propto \delta_{ij}$), otherwise we have $\rhob=\mbox{diag}(a,a,b,b)$ (with $a,b\neq 0$).

For $m=3$ we learn from Proposition~\ref{prop_n-3_text} that $\rhob=\mbox{diag}(a,a,b,0)$, while from section~\ref{subsec_n-3} we know that {$\Phi^{\{3\}}_{ijkl}=0=\Phi^{\{4\}}_{ijkl}$ and   $\Phi_{ij}\neq 0$ (Remark~\ref{rem_6d}).} Substituting the form of $\rhob$ into (\ref{eqn:algI})--\eqref{eqn:rhoC} {{and}} (\ref{eqn:algII}) further reveals that $\Phi_{ij}=\diag(\alpha,-\alpha,0,0)\neq 0$.

For $m=1$ one obviously has $\rhob=\mbox{diag}(a,0,0,0)\neq 0$. From section~\ref{subsec_m=1} we know that { $\Phi^{\{3\}}_{ijkl}=0=\Phi^{\{4\}}_{ijkl}$ and $\Phi_{ij}\neq 0$ (Remark~\ref{rem_6d_m=1}).}

Finally, for $m=2$ we can have two possibilities. If $\Phi^A_{ij}\neq0$ point \eqref{RT} of theorem~\ref{prop_GSHD} 
shows that $\rhob=\mbox{diag}(a,a,0,0)\neq 0$. On the other hand, if $\Phi^A_{ij}=0=\Phi^S_{ij}$ we still get $\rhob=\mbox{diag}(a,a,0,0)\neq 0$ from Proposition~\ref{prop_IIabd}, whereas if $\Phi^A_{ij}=0\neq\Phi^S_{ij}$ the more general form $\rhob=\mbox{diag}(a,b,0,0)$ is permitted (with $a\neq 0\neq b$).

We have checked in the various permitted cases  that the constraints (\ref{eqn:algI})--\eqref{eqn:rhoC} and \eqref{alg_no_Phi4}--\eqref{eqn:algII_N}
 allow for nontrivial type II Weyl tensor. {These results are summarized in table~\ref{tab_6D}, together with a few examples. Examples of spacetimes with distinct non-zero eigenvalues follow from the general discussion of section~\ref{subsubsec_examples_distinct}, which need not be repeated here.}

\begin{table}[t]
  \begin{center}
  \begin{tabular}{|c|l|c|c|c|l|c|}
    \hline Case & $m$ & Possible form of $\rhob$ &  O.C. & O.S. & I.T.G. & Examples  \\\hline
    $\Phi_{ij} \not=0$ & 4 & $\diag(a,a,a,a)$    & $\surd$ & $\surd *$ & ${\cal D}_{23}*$, ${\cal D}_{45}*$ & RT \\
     $  $ & 3 & $\diag(a,a,b,0)$     & if $b=a$ & X & ${\cal D}_{23}$ (if $b\neq a$) & ?($\dag$) \\
     $  $ & 2 \ ($\Phi^A_{ij}=0$) & $\diag(a,b,0,0)$    & if $b=a$ &  & ${\cal D}_{45}$ (${\cal D}_{23}$ if $b=a$) & {dS$_4\times$S$^2$} \\
     $  $ & 2 \ ($\Phi^A_{ij}\neq 0$) & $\diag(a,a,0, 0)$    & $\surd$ & $\surd$ & ${\cal D}_{23}$, ${\cal D}_{45}$ & ? \\
     $  $ & 1 & $\diag(a,0,0,0)$    & $\surd$ & X & & {dS$_3\times$S$^3$} \\
     $  $ & 0 & $\diag(0,0,0,0)$   & $\surd$ & &  & Kundt \\\hline
     $\Phi_{ij} =0  $ & 4 & $\diag(a,a,b,b)$    & if $b=a$ & $\surd$ &  ${\cal D}_{23}$, ${\cal D}_{45}$ (if $b\neq a$)& ?($\dag$) \\
     $    $ & 3 & X    & & & & X \\
     $ $  & 2 & $\diag(a,a,0,0)$  & $\surd$ &  $\surd$ & ${\cal D}_{23}$, ${\cal D}_{45}$ & ? \\
     $ $  & 1 & X     & & & & X \\
     $ $  & 0 & $\diag(0,0,0,0)$  & $\surd$ &  & & Kundt \\
\hline
  \end{tabular}
    \caption{Permitted forms of the optical matrix associated with a non-twisting mWAND in a six-dimensional type II spacetime in the two cases $\Phi_{ij} \not=0$  and $\Phi_{ij}=0$ (subtype II(abd)). Recall that $\Phi^A_{ij}\neq 0$ is possible only for $m=2$. A question mark indicates subcases that are in principle permitted but for which no examples are known (so we do not claim that these classes are non-empty), while  X corresponds to forbidden subcases. 
Although we are not aware of examples with $m=3$ and $\Phi_{ij} \not=0(=\Phi^A_{ij})$ having the most general structure, special solutions with $a=b$ (e.g., a static black string) can be obtained by warping once a type II 5d Einstein metric. {Similarly, generic $m=4$ solutions with $\Phi_{ij}=0$ are not known, but in the special case $a=b$ they} exist within the \RT family \cite{PodOrt06} (the ``$\mu=0$'' subcase). {These special examples are denoted by a $\dag$ in the table.} In the O.C. column we have indicated if/under what conditions $\rhob$ satisfies the optical constraint. In the O.S. column whether ${\cal D}=\mbox{Span}\{\mbb{2}+i\mbb{3},\mbb{4}+i\mbb{5},\bl\}$ defines an optical structure and in the I.T.G. column we listed integrable complex two-dimensional distribution that admit totally geodesic integral surfaces (cf. section~\ref{subsec_integrability} and appendix~\ref{app_shearfree}). The symbol $*$ means that the corresponding statement is true in a special subcase specified in the text (see section~\ref{subsec_integrability}). 
\label{tab_6D}}
  \end{center}
\end{table}

\subsection{Additional examples}

In addition to the examples mentioned in table~\ref{tab_6D} and in section~\ref{subsubsec_examples_distinct}, here we give  two more solutions with $m=1$ in six dimensions (and, necessarily, $\Phi_{ij}\neq0$). 

\subsubsection{Example with $\Lambda\neq0$}

We can use, e.g., AdS$_4\times$H$^2$ in the form 
\be
 \d s^2=\frac{3\beta^2}{z^2}(-\d t^2+\d\rho^2+\rho^2\d\phi^2+\d z^2)+\frac{\beta^2}{x^2}(\d x^2+\d y^2) .
\ee 
For the null congruence
\be
 \ell_a\d x^a=\d t+\d\rho ,
\ee 
with the spacelike frame vectors $\bbm_{(2)}=z\rho^{-1}/(\sqrt{3}\beta)\pa_\phi$, $\bbm_{(3)}=z/(\sqrt{3}\beta)\pa_z$, $\bbm_{(4)}=x\beta^{-1}\pa_x$, $\bbm_{(5)}=x\beta^{-1}\pa_y$,  one finds $\rho_{ij}=\mbox{diag}(a,0,0,0)$, where $a=z^2/(3\beta^2\rho)$.

\subsubsection{Example with $\Lambda$ arbitrary}

One can alternatively construct an example with {$\rho_{ij}=\mbox{diag}(a,0,0,0)$} by starting from a null geodesic congruence in dS$_3\times$S$^2$ and adding an extra dimension by Brinkmann's warp. For example, 
\be
 \d s^2=f(z)\left[-V(r)\d t^2+\frac{1}{V(r)}\d r^2+r^2\d\chi^2+\beta^2(\d\theta^2+\sin^2\theta\d\phi^2)\right]+\frac{\d z^2}{f(z)}, 
\ee 
with the null congruence
\be
 \ell_a\d x^a=\d t+V^{-1}\d r ,
\ee 
where  
\be 
V(r)=1-\frac{r^2}{2\beta^2}, \qquad f(z)=-\lambda z^2+2dz+b , \qquad d^2=\frac{1}{4\beta^2}-\lambda b.
\ee
Using the spacelike frame vectors $\bbm_{(2)}=f^{-1/2}r^{-1}\pa_\chi$, $\bbm_{(3)}=f^{-1/2}\beta^{-1}\pa_\theta$, $\bbm_{(4)}=f^{-1/2}\beta^{-1}\sin^{-1}\theta\pa_\phi$, $\bbm_{(5)}=f^{1/2}\pa_z$,  one finds $\rho_{ij}=\mbox{diag}(a,0,0,0)$, where $a=\frac{1}{rf}$. 

It follows from \cite{PraPraOrt07,OrtPraPra11} that $\Phi_{ij}\neq0$, {as it should be}. In general $\Lambda$ can have any sign but, in particular, we can set $\Lambda=0$ by choosing the parameter $\lambda=0$ in $f(z)$.

\subsection{Optical constraint, optical structures and totally geodesic complex two-dimensional null surfaces}

\label{subsec_integrability}

First, in the non-twisting case the optical constraint \eqref{canformL} is satisfied iff the optical matrix $\rhob$ possesses only one (possibly repeated) non-zero eigenvalue plus, possibly, some vanishing eigenvalues. This explains column O.C. (``optical constraint'') of table~\ref{tab_6D}.

Next, when $\bl$ is non-twisting one can always take an eigenframe of $\rhob$ to be parallelly transported \cite{OrtPraPra10,PraPra08} and thus in particular set $\M{i}{j}{0}=0$. It is then obvious that to each pair of repeated eigenvalues of $\rhob$ there corresponds an integrable two-dimensional totally null distribution, defined by two unit vectors of the corresponding eigenspace and by $\bl$ (cf. eq.~\eqref{D23_6d}). 

Further conditions on the Ricci rotation coefficients can be obtained by considering the Ricci identity (11k,\cite{OrtPraPra07}), i.e., eq.~\eqref{11k}, similarly as in section~\ref{sec_totgeod}.
Consequences of this equation enable one to prove the statements in columns  O.S. (``optical structure'') and I.T.G. (``integrable totally geodesic'' complex two-dimensional subspace) of table~\ref{tab_6D}. Since this is rather straightforward, let us just exemplify it to show that when $\rho_{ij}=\diag(a,a,b,b)$ (with $b\neq a$) the maximally totally  null distribution  ${\cal D}=\mbox{Span}\{\mbb{2}+i\mbb{3},\mbb{4}+i\mbb{5},\bl\}$ defines an optical structure (cf. eqs.~\eqref{6D_OS}). Namely, from \eqref{11k} with $i,j,k=224$ and $i,j,k=334$ one gets $\M{4}{2}{2}=\M{4}{3}{3}$, while $i,j,k=225$ and $i,j,k=335$ give $\M{5}{2}{2}=\M{5}{3}{3}$. Similarly, by taking $i,j,k=442/552$ and $i,j,k=443/553$ one obtains, respectively, $\M{2}{4}{4}=\M{2}{5}{5}$ and $\M{3}{4}{4}=\M{3}{5}{5}$. Next, with $i,j,k=324,234,325,235$ one finds $\M{3}{4}{2}=\M{2}{4}{3}=\M{3}{5}{2}=\M{2}{5}{3}=0$, while $i,j,k=524,425,534,435$ give $\M{5}{2}{4}=\M{4}{2}{5}=\M{5}{3}{4}=\M{4}{3}{5}=0$. This implies that eqs.~\eqref{6D_OS} are satisfied, as we wanted to prove. 

Let us finally remark that in the Robinson-Trautman case, corresponding to the first row of table~\ref{tab_6D}, eq.~\eqref{11k} reduces to $L_{1i}=0$  
(cf. also the results of \cite{PodOrt06,PraPra08}) and thus does not constraint the $\M{i}{j}{k}$. However, in the special subcase (cf.~\cite{PodOrt06}) in which the transverse space is of constant curvature one can show by inspection that there exists a natural parallelly transported frame such that ${\cal D}$ defines an optical structure and ${\cal D}_{23}$, ${\cal D}_{45}$ correspond to totally geodesic null two-surfaces.\footnote{This is a straightforward extension of the discussion in the 5d case, see footnote~10 of \cite{Ortaggioetal12}.} {{In fact the real $\mbb{i}$ can be paired arbitrarily in this case, giving rise also to other optical structures and totally geodesic null two-surfaces, and since }} there is a  continuous freedom in  spins there exist {{actually}} infinitely many optical structures. 
This partly answers a question raised in remark~5.3 of \cite{Taghavi-Chabert11b}, by showing that also in even dimensions a spacetime may admit more than $2^{[n/2]}$ optical structures (noting also that in this example replacing $\bl\leftrightarrow\bn$ gives other optical structures).

\subsection{Type III/N spacetimes}

By the results of \cite{Pravdaetal04}, in this case there are only two admitted structures of $\rhob$: either $\rhob=0$ (Kundt spacetimes) or $\rhob=\diag(a,a,0,0)\neq0$. Since eq.~\eqref{11k} does not contain any Weyl components of zero boost weight, the above discussion still applies and implies that in the case $\rhob=\diag(a,a,0,0)\neq0$ both ${\cal D}_{23}$ and ${\cal D}_{45}$ correspond to totally geodesic null two-surfaces and that ${\cal D}$ defines an optical structure -- explicit examples can be constructed as explained in \cite{OrtPraPra10}.

\section*{Acknowledgments}

We are grateful to Lode Wylleman for useful discussions. The authors acknowledge support from research plan {RVO: 67985840} and research grant no P203/10/0749.

\renewcommand{\thesection}{\Alph{section}}
\setcounter{section}{0}

\renewcommand{\theequation}{{\thesection}\arabic{equation}}

\section{Algebraic constraints for type II spacetimes}
\setcounter{equation}{0}

\label{app_alg_eq}

\subsection{{General mWAND}}

For completeness and future reference, let us give here some additional algebraic constraints following from the Bianchi and Ricci identities for type II Einstein spacetimes in arbitrary dimensions. 

Note that rhs of eq. \eqref{A4} is not symmetric under $ij\leftrightarrow kl$ and thus \eqref{A4}$_{ijkl}-$\eqref{A4}$_{klij}$ leads to an  algebraic equation
\be
2(\Phia_{ij} A_{kl}-\Phia_{kl} A_{ij}) +\Phi_{k[j}\rho_{i]l}-\Phi_{l[j}\rho_{i]k}-\Phi_{j[k}\rho_{l]i}+\Phi_{i[k}\rho_{l]j}
+\Phi_{ij[k|s}\rho_{s|l]}-\Phi_{kl[i|s}\rho_{s|j]}=0 , \label{alg_no_Phi4}
\ee
and its trace gives the antisymmetric part of \eqref{B8}.
By multiplying by $\rho_{lj}$ and using \eqref{B8}--\eqref{B8asym} and symmetries of the Weyl tensor, after  lengthy calculations one arrives at an equation with no $\Phi_{ijkl}$ terms
\be
\rho \Phi_{j[i|}\rho_{j|k]}+\Phi_{jl}\rho_{jl}A_{ki}+\rho_{jl}\rho_{jl}\Phia_{ik}
+\Phi_{jl}\left[ 3\rho_{[j|i}\rho_{|l]k} +\rho_{[k|l}\rho_{j|i]} \right]
+2\rho_{jl}\left[ \Phi_{j[i|}A_{l|k]}+2\Phia_{l[i|}\rho_{j|k]}\right]=0.
 \label{constr_twist}
\ee

\subsection{{Non-twisting mWAND with $\Phia_{ij}=0$}}

As already mentioned, by differentiating algebraic constraints and using the Bianchi and Ricci identities one can arrive at further algebraic constraints.
For example, in the {\em non-twisting} case {with $\Phia_{ij}=0$},\footnote{{{Recall that in the non-twisting case with $\Phia_{ij}\neq0$ the form of $\rhob$ is fully determined by point \eqref{PhiA} of theorem~\ref{prop_GSHD}, cf.~\cite{Ortaggioetal12}.}}}
by differentiating~\eqref{eqn:algI} 
we arrive at \eqref{eqn:rhoC} {(equivalent to setting $A_{ij}=0=\Phia_{ij}$ in \eqref{thornB8}), which can be rewritten as}
\be
 - [\rhob^2] \Phi_{ik} + \Phi (\rhob^2)_{ik} +  2 \Phi_{ij} (\rhob^2)_{jk} + {\Phi}_{ijkl} {(\rhob^2)}_{jl}  =0 \qquad (A_{ij}=0{=\Phia_{ij}}),
 \label{eqn:algII_0}
\ee
where $[\rhob^2]$ denotes the trace of {$\rhob^2$}, etc. The trace of \eqref{eqn:algII_0} vanishes and multiplying \eqref{eqn:algII_0} by $\rho_{ik}$ gives
\be
	[\rhob^2][\Phib \cdot \rhob]-[\rhob][\Phib \cdot\rhob^2]=0 \qquad {(A_{ij}=0{=\Phia_{ij}})} .
	\label{trace}
\ee
Differentiation {of \eqref{eqn:algII_0}} leads to 
\bea
 & & -3 \left[ - [\rhob^3] \Phi_{ik} + \Phi (\rhob^3)_{ik} +  2 \Phi_{ij} (\rhob^3)_{jk} + {\Phi}_{ijkl} {(\rhob^3)}_{jl}   \right] \nonumber\\
		& & {}+ \left( \Phi [\rhob^2] + [\Phib \cdot\rhob^2] 
		\right) \rho_{ik} 
- \left(\Phi \rho + [\Phib \cdot \rhob] 
\right) (\rhob^2)_{ik}=0 \qquad (A_{ij}=0{=\Phia_{ij}})
\label{eqn:algII} ,
\eea
with trace {of \eqref{eqn:algII}} vanishing  due to \eqref{trace}. 

For $\Phi_{ij}=0$, the $(N-1)$th derivative of  \eqref{B8}   reads
\be
	{\Phi}_{ijkl} {(\rhob^N)}_{jl} =0, \quad N =1,2,\dots \qquad (A_{ij}=0{=\Phi_{ij}}).
	\label{eqn:algII_N}
\ee

\section{{Non-twisting type III Einstein spacetimes in arbitrary dimension}}

\setcounter{equation}{0}

\label{app_type_N_III}

For non-twisting Einstein spacetimes of type N, $\rho_{ij}=\diag({a},{a},0,\dots,0)$ \cite{Pravdaetal04,Durkeeetal10}, where ${a}$ {vanishes} in the case of Kundt spacetimes. Let us present  here a new derivation of the form of the optical matrix for a non-twisting type III Einstein spacetime in arbitrary
dimension. This is considerably simpler than the original derivation of \cite{Pravdaetal04}. 

For non-twisting type III, eq. (59) of \cite{Pravdaetal04} reduces to {(in the notation of \cite{Durkeeetal10})}
\be
2 \Psi'_{lij} \rho_{lk} = \rho \Psi'_{kij}.   \label{typeIIInon-twist}
\ee
In the Kundt case ($\rhob=0$) this becomes an identity, while for non-twisting non-Kundt Einstein spacetimes (to which we restrict from now on) necessarily $\rho\neq0$ (cf. Proposition~1 of \cite{OrtPraPra07}) and \eqref{typeIIInon-twist} gives non-trivial information.
Following \cite{OrtPraPra10} we choose a parallelly propagated frame with diagonal form of  $\rho_{ij}=\diag(\rho_2,\rho_3,\dots)$.
Without loss of generality we can assume that $\Psi'_{2ij} \not=0$ for some values of $i,\ j$. Eq.~\eqref{typeIIInon-twist} then implies
\be
\rho_2 = \frac{\rho}{2}{\neq0} \label{typeIIIrho}.
\ee
Using the $r$-dependence of the non-vanishing eigenvalues of $\rho_{ij}$ (see eq.~\eqref{rho_i}), this gives
\be
\frac{2}{r-b_2} = \sum_p \frac{1}{r-b_p},
\ee
where $p$ corresponds to all non-vanishing eigenvalues of $\rho_{ij}$.
Apart from the Kundt case with $\rho_{ij}=0$, this equation can be satisfied only for $\rho_{ij}$ of rank 2, with ${a\equiv}\rho_2=\rho_3$.
Thus {\em the optical matrix for the non-twisting type III Einstein spacetime is either vanishing (Kundt) or 
$\rho_{ij}=\diag({a},{a},0,\dots,0)$}, exactly as in the type N case.

\section{On the geodetic and shearfree condition in four and higher dimensions}

\setcounter{equation}{0}

\label{app_shearfree}

\subsection{Goldberg-Sachs in four dimensions}

\label{app_GS_4D}

The four-dimensional Goldberg-Sachs theorem is usually expressed as a statement about a null congruence being geodesic ($\kappa=0$) and shearfree ($\sigma=0$). As well-known, this property can be interpreted as different but equivalent geometrical statements, which we now briefly review. This will later be useful for discussing the higher dimensional case.

\subsubsection{Integrability of two-spaces}

\label{subsubsec_integr_d4}

By considering the commutator $[\delta,D]$ (see, e.g., p.~77 of \cite{Stephanibook})
one immediately sees that  $\bl$ is geodesic and shearfree if and only if the distribution
\be
 {\cal D}=\mbox{Span}\{\bbm,\bl\} 
\ee
is integrable (cf.~\cite{penrosebook2}), i.e. $[{\cal D},{\cal D}]\subset{\cal D}$.
A similar conclusion holds of course for $\bar{\cal D}$. The geodesic{\&}shearfree property can thus be seen as a statement about the null bivector $\bl\wedge\bbm$ instead of one about the vector field $\bl$. While the vector field $\bl$ defines a privileged real null direction lying in ${\cal D}\cap\bar{\cal D}$, $\bbm$ is not fixed uniquely (null rotations with $\bl$ fixed, spins and boost can be performed, under which $\bl\wedge\bbm$ is fixed, up to a rescaling, and $\kappa=0=\sigma$ is preserved). The significance of this viewpoint appears clearly when considering complex extensions of the Goldberg-Sachs theorem \cite{PlebHac75}.

\subsubsection{Conformal structure of the screen space}

Given a generic null vector field $\bl$, its flow preserves the conformal structure of the ``screen space'' $L^\bot/L$ when \cite{RobTra83} (see also \cite{NurTra02,Trautman02a,Trautman02b})
\be 
 {\cal L}_\l g=\rho g+\l\otimes\xi+\xi\otimes\l  \quad \Leftrightarrow  \quad \kappa=0=\sigma .
 \label{Lie_4}
\ee

In four dimensions only, this is equivalent to requiring that the complex structure of the screen space is preserved by the flow of $\bl$ \cite{RobTra83,NurTra02,Trautman02a,MasTag08}.

This condition can also be derived by imposing Maxwell equations on a null field (two-form) \cite{robinsonnull} (cf. also, e.g., \cite{Stephanibook,penrosebook2,NurTra02,Trautman02a,Trautman02b}).

\subsection{Conformal structure of the screen space: the standard geodesic{\&}shearfree condition}

\label{subsec_screenHD}

Ref.~\cite{RobTra83} (see also \cite{Trautman02b}) proposed to extend (\ref{Lie_4}) to higher dimensions as a definition of the geodesic{\&}shearfree condition. As it turns out,
\be 
 {\cal L}_\l g=\rho g+\l\otimes\xi+\xi\otimes\l  \quad \Leftrightarrow  \quad \kappa_{i}=0=S_{ij}-\frac{S_{kk}}{n-2}\delta_{ij} ,
 \label{Lie}
\ee
so that indeed $\bl$ is geodesic and shearfree in the standard sense. Some comments on algebraically special Einstein spacetimes admitting such a congruence have been already given in section~\ref{subsec_OC} and appendix~\ref{app_shearfreetwist}.

Note that for $n>4$ eq.~(\ref{Lie}) does {\em not} follow from Maxwell's equations for a null field (intended again as a two-form), while geodeticity {($\kappa_{i}=0$)} does, and shear is non-zero for expanding solutions \cite{Ortaggio07,Durkeeetal10}.

\subsection{Existence of an optical structure}

\label{subsubsec_optical}

\subsubsection{Complex notation}

For later discussions it will be useful to use a complex basis. Namely, in {\em even} dimensions one can take the frame $(\bl,\bn,\bmu_A,\bar\bmu_A)$ and in {\em odd} dimensions the frame $(\bl,\bn,\bmu_A,\bar\bmu_A,\bx)$. Except for the unit spacelike vector $\bx$, all the vectors are null, the complex $\bmu_A$ are defined by
\be
  \bmu_2=\frac{1}{\sqrt{2}}(\mbb{2}+i\mbb{3}) , \qquad \bmu_4=\frac{1}{\sqrt{2}}(\mbb{4}+i\mbb{5}) , \qquad \ldots ,
  \label{mu}
\ee
and ${\bar\bmu_A}$ by their complex conjugates, where $A=2\mu$, $\mu=1,\ldots,(n-2-\epsilon)/2$, with $\epsilon=0,1$ for even and odd dimensions, respectively (in 4d this would be the standard NP frame). The metric becomes
\be
 g=\bl\otimes \bn+\bn\otimes\bl+\bmu_A\otimes\bar\bmu_A+\bar\bmu_A\otimes\bmu_A+\epsilon \bx\otimes\bx .
 \label{complex_g}
\ee

Then one can define the complex counterpart of the Ricci rotation coefficients. These are defined as an obvious extension of the real coefficients, e.g.,
\be
 \cL_{AB}=\mu^a_A\mu^b_B\nabla_b l_{a} , \qquad \cL_{A\bar B}=\mu^a_A\bar\mu^b_B \nabla_b l_{a} , \qquad \cM{A}{C}{B}=\mu^a_C\mu^b_{{B}}\nabla_b \mu_{A a}, 
\ee
and other coefficients (and their complex-conjugates) are defined similarly.

\subsubsection{Optical structure}

\label{subsubsec_OS}

Ref.~\cite{HugMas88} studied the consequences of Maxwell's equations for a null field defined as an $n/2$-form in $n$ even dimensions. From these, they arrived at a generalization of the geodesic{\&}shearfree condition different from the $\kappa=0=\sigma$ condition discussed in section~\ref{subsec_screenHD} (see also \cite{NurTra02,Trautman02a,Trautman02b,MasTag08}). Recently, this has been extended also to odd dimensions \cite{Taghavi-Chabert11}. Namely, consider the totally null  $(n-\epsilon)/2$-dimensional distribution (recall that $A$ in $\mu_A$ can be only even)
\be
 {\cal D}=\mbox{Span}\{\bl,\bmu_2,\ldots,\bmu_{(n-2-\epsilon)}\} .
 \label{D_gen}
\ee
If ${\cal D}$ and ${\cal D}^\bot$ are integrable, i.e.,
\be
 [{\cal D},{\cal D}]\subset{\cal D} , \qquad [{\cal D}^\bot,{\cal D}^\bot]\subset{\cal D}^\bot ,
 \label{Robinson}
\ee
${\cal D}$ is said to define an ``{\em optical structure}'' \cite{Taghavi-Chabert11} (note that ${\cal D}^\bot={\cal D}$ in even dimensions).  In 4d, eq.~\eqref{Robinson} reduces to the standard conditions $\kappa=0=\sigma$ (as discussed in section~\ref{subsubsec_integr_d4}), which indeed corresponds to the conditions coming from the Mariot-Robinson theorem \cite{Stephanibook,penrosebook2} for a Maxwell null two-form. In  higher dimensions, using the complex counterpart of the commutators given in \cite{Coleyetal04vsi}, one finds that (\ref{Robinson}) is equivalent to 
\be
 \cL_{A0}=\cL_{x0}=\cL_{AB}=\cL_{Ax}=\cL_{xA}=\cM{A}{B}{0}=\cM{A}{x}{0}=\cM{A}{B}{C}=\cM{A}{B}{x}=\cM{x}{A}{B}=0 .
 \label{optical_complex}
\ee
(In fact one first obtains conditions such as $\cM{A}{B}{C}=\cM{C}{B}{A}$, but $\cM{A}{B}{C}=-\cM{B}{A}{C}$, etc.) The equations above represent the conditions obtained when $n$ is odd, however, if $n$ is even one can simply drop all equations containing $x$ (this will be understood from now on).

Note, in particular, that the vector $\bl$ is {\em geodesic} ($\cL_{A0}=\cL_{x0}=0$), but generally is not required to be shearfree (except when $n=4$). 
In even dimensions it has also been observed that (\ref{Robinson}) means that the complex structure of the screen space is preserved \cite{NurTra02,Trautman02a,MasTag08}. We have shown in \cite{Ortaggioetal12} that in 5d a very large class of algebraically special Einstein spacetimes possesses an optical structure. In appendix \ref{app_OS5d} we extend the results of \cite{Ortaggioetal12} by showing that, in fact, all algebraically special 5d Einstein spacetimes possess (at least) one optical structure.

\paragraph{Optical structure in six dimension} Let us rewrite the conditions \eqref{optical_complex} in real notation in the special case $n=6$, which is useful for the discussion in the main text (see \cite{Ortaggioetal12} for the case $n=5$). One readily gets
\beqn
 & & \kappa_{i}=0 \quad (i=2,\ldots,5) , \nonumber \\
 & & \rho_{22}=\rho_{33}, \quad \rho_{23}=-\rho_{32}, \quad \rho_{44}=\rho_{55}, \quad \rho_{45}=-\rho_{54},  \nonumber \\
 & & \rho_{24}=\rho_{35}, \quad \rho_{42}=\rho_{53}, \quad \rho_{34}=-\rho_{25}, \quad \rho_{43}=-\rho_{52}, \nonumber \\
 & & \M{2}{4}{0}=\M{3}{5}{0}, \quad \M{3}{4}{0}=-\M{2}{5}{0}, \qquad\qquad\qquad\qquad\qquad\qquad\qquad (n=6) \label{6D_OS} \\
 & & \M{2}{4}{2}-\M{3}{4}{3}=\M{2}{5}{3}+\M{3}{5}{2}, \quad -\M{2}{5}{2}+\M{3}{5}{3}=\M{2}{4}{3}+\M{3}{4}{2}, \nonumber \\ 
 & & \M{2}{4}{4}-\M{2}{5}{5}=\M{3}{4}{5}+\M{3}{5}{4}, \quad \M{2}{4}{5}+\M{2}{5}{4}=-\M{3}{4}{4}+\M{3}{5}{5} . \nonumber 
\eeqn

In particular, if $\bl$ is twistfree then $\rhob$ has two pairs of repeated eigenvalues.

\subsection{Optical constraint}

Based on results for Kerr-Schild spacetimes, Ref.~\cite{OrtPraPra09} put forward yet another possible generalization of the shearfree condition for mWANDs in higher dimensions. This is the so called ``optical constraint'', already discussed in section~\ref{subsec_OC}, which involves only the null direction $\bl$ (as opposed to the optical structure discussed above). The Lorentz transformation freedom of null rotations preserving $\bl$, boosts and spins is thus retained in this case (indeed spins can be used to arrive at the canonical form~\eqref{canformL}), see also \cite{OrtPraPra10} for related comments. We have shown in \cite{Ortaggioetal12} that in 5d a very large class of algebraically special Einstein spacetimes admits an mWAND obeying the optical constraint, {{and Ref.~\cite{Wylleman_priv} extended our result to prove that in fact {\em all} algebraically special Einstein spacetimes admit such an mWAND} (see also \cite{OrtPraPra12rev})}.
Note that in 4d the optical constraint is a necessary condition for $\bl$ to be a repeated principal null direction but is not {sufficient} \cite{Ortaggioetal12}.

\subsection{Integrability of a (complex) two-dimensional totally null distribution}  

\subsubsection{${\cal D}_{23}$ integrable}

As a further generalization of the geodesic{\&}shearfree condition one can consider the integrability of the complex two-dimensional totally null distribution
\be 
 {\cal D}_{23}=\mbox{Span}\{\bmu_{2},\bl\} .
 \label{D}
\ee
It is easy to show that ${\cal D}_{23}$ is integrable if and only if 
\be
 \cL_{20}=0, \qquad \cL_{B2}=\cM{2}{B}{0}, \qquad \cL_{\bar B2}=\cM{2}{\bar B}{0} \quad (B\neq 2) , \qquad \cL_{x2}=\cM{2}{x}{0} .
 \label{2integr}
\ee
In 4d this again reduces to the standard $\kappa=0=\sigma$ condition.

In particular, in six dimensions \eqref{2integr} can be rewritten in real notation as
\beqn
 & & \kappa_{2}=0=\kappa_{3} , \qquad  \rho_{22}=\rho_{33}, \quad \rho_{23}=-\rho_{32}, \nonumber \label{D23_6d} \\
 & & \rho_{42}=\M{2}{4}{0}, \quad \rho_{43}=\M{3}{4}{0}, \quad \rho_{52}=\M{2}{5}{0}, \quad \rho_{53}=\M{3}{5}{0}  \qquad\qquad (n=6) . 
\eeqn

\subsubsection{${\cal D}_{23}$ integrable with totally geodesic integral surfaces}  

We can strengthen the above conditions by further requiring the integral surfaces of ${\cal D}_{23}$ to be {\em totally geodesic}. The corresponding equations read
\beqn
 & & \cL_{B0}=0=\cL_{x0}, \qquad \cL_{B2}=0=\cL_{x2}, \qquad \cL_{\bar B2}=0 \quad (B\neq 2) , \nonumber \\
 & & \cM{2}{B}{0}=0=\cM{2}{x}{0}, \qquad \cM{2}{B}{2}=0=\cM{2}{x}{2}, \qquad \cM{2}{\bar B}{0}=0 \quad (B\neq 2) , \qquad \cM{2}{\bar B}{2}=0 \quad (B\neq 2) .
 \label{2totgeod}
\eeqn
Now $\bl$ is necessarily geodesic ($\cL_{B0}=0=\cL_{x0}$). In 4d \eqref{2totgeod} is equivalent to \eqref{2integr} because $B=2$ is the only possibility (and there are no $x$-components). In real notation \eqref{2totgeod} can be rewritten as (where we define $\hat k\neq 2,3$)
\beqn
 & & \kappa_{i}=0 , \qquad  \rho_{22}=\rho_{33}, \quad \rho_{23}=-\rho_{32}, \qquad \rho_{\hat k2}=0=\rho_{\hat k3} \nonumber \\
 & & \M{2}{\hat k}{0}=0=\M{3}{\hat k}{0} , \label{2totgeod_real} \\
 & & \M{2}{\hat k}{2}=\M{3}{\hat k}{3}, \quad  \M{2}{\hat k}{3}=-\M{3}{\hat k}{2} \nonumber .
\eeqn

\section{Optical structures in five dimensions}

\label{app_OS5d}

Proposition~4 of \cite{Ortaggioetal12} gives a set of sufficient conditions for a five-dimensional Einstein spacetime of type II or more special to possess an {\em optical structure} \cite{Taghavi-Chabert11}. In particular, it shows that, except possibly for Kundt spacetimes (and for a special subclass ``(iii)'' of genuine type II, later proven not to exist \cite{Wylleman_priv}), all algebraically special Einstein spacetimes admit an optical structure in five dimensions. It is the purpose of this  appendix to show that this in fact holds also for Kundt spacetimes. (Obviously, if a complex optical structure is integrable its complex conjugate is integrable too and this will be understood in the following.) Combining this with the proof of \cite{Wylleman_priv} that the special subclass (iii) of genuine type II is empty, we arrive at 
\begin{prop}
\label{prop_integrab_5D}
In a five-dimensional Einstein spacetime admitting a multiple WAND~$\lb$ there always exists an optical structure. In the case of type D spacetimes there exist in fact (at least) two optical structures. 
\end{prop}

\begin{proof}
Let us  show that Einstein spacetimes of the Kundt class always possess an optical structure in five dimensions. In other words, we need to show that in such spacetimes there always exists a null frame $\{\bl,\bn,\mbb{2},\mbb{3},\mbb{4}\}$ such that 
the totally null distribution
\be
 {\cal D}=\mbox{Span}\{\mbb{2}+i\mbb{3},\bl\} ,
\ee
and its orthogonal complement 
\be
 {\cal D}^\bot=\mbox{Span}\{\mbb{2}+i\mbb{3},\mbb{4},\bl\} ,
\ee 
are both integrable. This is equivalent to \cite{Ortaggioetal12} (see also appendix~\ref{app_shearfree}, and 
{\eqref{L1i_M}} for the definition of $\M{a}{b}{c}$)
\beqn
 & & \kappa_{i}=0, \qquad \rho_{33}=\rho_{22}, \qquad \rho_{32}=-\rho_{23},\qquad \rho_{24}=0=\rho_{34}, \qquad \rho_{42}=0=\rho_{43},  			  \label{Dorth_integr_k_rho} 
 \\  & & \M{2}{4}{0}=0=\M{3}{4}{0} , \qquad \M{2}{4}{2}=\M{3}{4}{3}, \qquad \M{2}{4}{3}=-\M{3}{4}{2} .
 \label{Dorth_integr_M} 
\eeqn

Kundt spacetimes admit a metric in the form \cite{Coleyetal03,ColHerPel06,PodZof09} 
\be
 \d s^2 =2\d u\left[\d r+H(u,r,x)\d u+W_\alpha(u,r,x)\d x^\alpha\right]+ g_{\alpha\beta}(u,x) \d x^\alpha\d x^\beta , \label{Kundt_gen}
\ee
where $\alpha,\beta=2,3,4$ in five dimensions. Here $\bl=\partial_r$ is a geodesic, twistfree, shearfree, non-expanding mWAND, so that \eqref{Dorth_integr_k_rho} is automatically satisfied. Now, define a null frame with $\ell_a\d x^a=\d u$, $n_a\d x^a=\d r+H\d u+W_\alpha\d x^\alpha$ and the spacelike vectors $\mbb{i} $ living in the three-dimensional transverse (Euclidean) space spanned by the $x^\alpha$ (their components will be, in particular, independent of $r$). Then one immediately finds
\be
 \M{i}{j}{0}=0, 
\ee
so that the first of \eqref{Dorth_integr_M} is satisfied in this frame. 

Next, using the Christoffel symbols given in \cite{PodZof09} one can check that  
\be
 m_{(i)\a;\beta}=m_{(i)\a||\beta} ,
\ee
where the covariant derivatives on the r.h.s. is taken w.r.t. the transverse metric $g_{\a\beta}$. 

It follows that
\be
 {\M{i}{j}{k}} ={\Mt{i}{j}{k}} , 
\ee
where the connection coefficients on the r.h.s. are those computed w.r.t. the transverse metric $g_{\a\beta}$ (and the o.n. frame vectors of the transverse space are simply the ``projections'' $\tilde m_{(i)\a}$  of the $m_{(i)\a}$ vectors). 

One then finds
\be
 [\mbtt{2}+i\mbtt{3},\mbtt{4}]=(\Mt{4}{2}{4}+i\Mt{4}{3}{4})\mbtt{4} +[-\Mt{2}{4}{2}+i(\Mt{2}{3}{4}-\Mt{2}{4}{3})]\mbtt{2}+[\Mt{3}{2}{4}-\Mt{3}{4}{2}-i\Mt{3}{4}{3}]\mbtt{3} ,
\ee
so that $\mbox{Span}\{\mbtt{2}+i\mbtt{3},\mbtt{4}\}$ is integrable iff
\be
 \Mt{3}{4}{3}=\Mt{2}{4}{2} , \qquad \Mt{3}{4}{2}=-\Mt{2}{4}{3} .
 \label{integr_234}
\ee 
Proving the integrability of ${\cal D}$ and ${\cal D}^\bot$ is thus now reduced to proving the integrability of $\mbox{Span}\{\mbtt{2}+i\mbtt{3},\mbtt{4}\}$. The transverse frame is arbitrary, so we just need to show that there exists at least one frame satisfying this integrability property.

If we define the complex null vector field
\be
 \bmu=\frac{1}{\sqrt{2}}(\mbtt{2}+i\mbtt{3}) , 
 \label{def_mu}
\ee
using \eqref{integr_234} the required integrability condition $[\bmu,\mbtt{4}]=\alpha\bmu+\beta\mbtt{4}$ reads
\be
 \mu_{a||b}m_{(4)}^a\mu^b=0 .
\ee 
Thus by choosing a complex null geodesics $\bmu$ in the transverse three-space we automatically obtain  the integrability of the corresponding  distributions ${\cal D}$ and ${\cal D}^\bot$ and {(together with the result of \cite{Wylleman_priv})} our proof is complete.  
\end{proof}

\section{Shearfree twisting spacetimes (even dimensions)}

\label{app_shearfreetwist} 

Twisting geodesic mWANDs with zero shear are forbidden in odd dimensions \cite{OrtPraPra07} but they are permitted in {\em even} dimensions and they have necessarily $\det(\rhob)\neq 0$ (as can be easily seen in a frame adapted to $A_{ij}$, using the fact that $S_{ij}\propto\delta_{ij}$). Here we present an explicit example in six dimensions. {To our knowledge, this is the first such example that has been identified.} 

First, consider the six-dimensional Ricci flat Taub-NUT metric \cite{ManSte04} 
\beqn
 \d s^2= & & -F(r)(\d t-2n_1\cos\theta_1\d\phi_1-2n_2\cos\theta_2\d\phi_2)^2+\frac{\d r^2}{F(r)} \nonumber \\
 & & {}+(r^2+n_1^2)(\d\theta_1^2+\sin\theta_1^2\d\phi_1^2)+(r^2+n_2^2)(\d\theta_2^2+\sin\theta_2^2\d\phi_2^2) ,
 \label{MasSte6D}
\eeqn
where
\be
 F(r)=\frac{r^4/3+(n_1^2+n_2^2)r^2-2mr-n_1^2n_2^2}{(r^2+n_1^2)(r^2+n_2^2)} .
\ee
We observe that this is a spacetime of type D. A geodetic mWAND is given by
\be
 \ell_a\d x^a=\d t+F(r)^{-1}\d r-2n_1\cos\theta_1\d\phi_1-2n_2\cos\theta_2\d\phi_2 ,
\ee 
while a second one can simply be obtained by reflecting $\bl$ as $t\to-t$, $\phi_1\to-\phi_1$, $\phi_2\to-\phi_2$ \cite{PraPraOrt07}. Using the frame vectors
\beqn
 & & m_{(2)a}\d x^a=\sqrt{r^2+n_1^2}\d\theta_1, \qquad m_{(3)a}\d x^a=\sqrt{r^2+n_1^2}\sin\theta_1\d\phi_1, \nonumber \\
 & & m_{(4)a}\d x^a=\sqrt{r^2+n_2^2}\d\theta_2, \qquad m_{(5)a}\d x^a=\sqrt{r^2+n_2^2}\sin\theta_2\d\phi_2,  
\eeqn
one finds
\be
\rhob = \left(
\begin{array}{cccc} 
      \displaystyle \frac{r}{r^2+n_1^2} & \displaystyle -\frac{n_1}{r^2+n_1^2} & 0 & 0\\ 
			\displaystyle \frac{n_1}{r^2+n_1^2} & \displaystyle \frac{r}{r^2+n_1^2}  & 0 & 0 \\ 
			0 & 0 & \displaystyle \frac{r}{r^2+n_2^2} & \displaystyle -\frac{n_2}{r^2+n_2^2}\\ 
			0 & 0 & \displaystyle \frac{n_2}{r^2+n_2^2} & \displaystyle \frac{r}{r^2+n_2^2} 
  \end{array}
 \right) . 
\ee

One can easily check that $\rhob$  obeys the optical constraint~(\ref{OC2}). {Moreover, we also observe that in the spacetime~\eqref{MasSte6D} the maximally totally null distribution ${\cal D}=\mbox{Span}\{\mbb{2}+i\mbb{3},\mbb{4}+i\mbb{4},\bl\}$
defines an {\em optical structure}  (concept introduced and discussed in \cite{HugMas88,NurTra02,Trautman02a,Trautman02b,Taghavi-Chabert11,Taghavi-Chabert11b}), and both the totally null distributions ${\cal D}_{23}=\mbox{Span}\{\mbb{2}+i\mbb{3},\bl\}$ and ${\cal D}_{45}=\mbox{Span}\{\mbb{4}+i\mbb{5},\bl\}$ are integrable, with totally geodesic integral surfaces (see appendix~\ref{app_shearfree} for the corresponding definitions and conditions).}
The null vector field $\bl$ is generically shearing ($S_{22}=S_{33}\neq S_{44}=S_{55}$), however in the special case $n_1=n_2$ (corresponding to the solutions of \cite{AwaCha02}) it becomes {\em shearfree} (while still being expanding and twisting). The class of shearfree twisting spacetimes is thus non-empty in higher dimensions.

\section{{Non-degenerate, non-twisting geodesic mWANDs violating the optical constraint ($n\ge 7$)}}

\setcounter{equation}{0}

\label{app_violating}

In the main text we have seen examples of Einstein spacetimes admitting a non-twisting mWAND violating the optical constraint, see e.g. table~\ref{tab_distinct_eigenval}. However, all of them have a degenerate $\rhob$ (i.e., $m<n-2$). In this appendix we provide also some examples with $m=n-2$, relevant to theorem~\ref{prop_GSHD} (point \eqref{Phi0}) and proposition~\ref{prop_non_deg2} (point 2.).

\subsection{Metric}

Using the theory of conformal Einstein spaces \cite{Brinkmann25} (reviewed, e.g., in \cite{petrov}), by taking a double Brinkmann warp one can construct the following Einstein space  (satisfying $R_{ab}=2\Lambda g_{ab}/(n-2)$):
\be
 \d s^2=\lambda r^2\d u^2+2\d u\d r+(\lambda ur-1)^2\d\sigma_\lambda^2+r^2\d\Sigma_0^2 , 
 \label{general_violating}
\ee
where 
\be
 \lambda=2\frac{\Lambda}{(n-1)(n-2)} ,
\ee
$\d\sigma_\lambda^2$ is a Riemannian Einstein space of dimension $n_\sigma$ with Ricci scalar $R_\sigma=n_\sigma(n_\sigma-1)\lambda$ and $\d\Sigma_0^2$ a Ricci-flat Riemannian space of dimension $n_\Sigma$. In order for $\d s^2$ to have a non-zero Weyl tensor, the metrics $\d\sigma_\lambda^2$ or $\d\Sigma_0^2$ cannot be both of constant curvature, so that either $n_\sigma\ge 4$ or $n_\Sigma\ge 4$ (or both). The complete spacetime has thus dimension $n=2+n_\sigma+n_\Sigma\ge 7$ (unless it is of constant curvature, which is of no interest to us).

Let us define the vector field $\bl$ with covariant and contravariant components
\be
 \ell_a\d x^a=\d u , \qquad \ell^a\pa_a=\pa_ r ,
 \label{l_ex}
\ee
which is obviously  null and hypersurface-orthogonal, thus automatically geodesic and twistfree, and $r$ is an affine parameter along it. This can be accompanied by another null vector $\bn$ (also hypersurface-orthogonal)
\be
 n_a\d x^a=\d r+\frac{\lambda}{2}r^2\d u ,
 \label{n_ex}
\ee
which satisfies the normalization condition $\bl\cdot\bn=1$.

\subsection{Weyl tensor}

From the results of \cite{Brinkmann25} together with those on direct product spacetimes \cite{PraPraOrt07} it readily follows that
\be
 C_{uabc}=0=C_{rabc} .
 \label{C_ur_ex}
\ee 

Let us consider a set of $n-2$ spacelike o.n. vectors $\bbm_{(i)}$ orthogonal to $\bl$ and $\bn$, composed of a subset of vectors denoted by $i=A, B, \ldots$   living in the subspace of $\d\sigma_\lambda^2$ and a subset denoted by $i=I,J, \ldots$ living in the subspace of $\d\Sigma_0^2$. Then the only non-zero Weyl frame components are given by $C_{ABCD}$ and $C_{IJKL}$, and their $r$-dependence is (recall the first definition in \eqref{bw0}) 
\be
 \Phi_{ABCD}=\frac{1}{(\lambda ur-1)^2}C^0_{ABCD} , \qquad \Phi_{IJKL}=\frac{1}{r^2}C^0_{IJKL} , 
 \label{C_ABIJ_ex}
\ee
where quantities with superscript $^0$ do not depend on $r$ (these are in fact the Weyl components of the respective ``subspaces''). 

In particular, one has
\be
 \WD{ij} =0 .
\ee
Note also that one has generically $W_{ij}\neq 0$ (recall  definition \eqref{def_W} and constraint~\eqref{W_phi}).

\subsection{Multiple WAND(s) and optics}

From the above results it follows that both $\bl$ and $\bn$ are {double} WANDs and the spacetime is thus of type~D. In fact a bit more than that, since 
\be
C_{abcd}\ell^d=0=C_{abcd}n^d \label{Wandcond}
\ee
 (cf. \cite{PraPra05,Ortaggio09} for the meaning of such conditions in terms of the Weyl type)  and the type is D(abd).\footnote{The construction {of metric \eqref{general_violating}} is carried out more naturally in a different coordinate system (not adapted to $\bl$) in which $\bl$ and $\bn$ are related by a coordinate transformation (``time-reflection'') that leaves the metric invariant: therefore they share the same geometric properties.}

The optical matrix $\rhob$ of $\bl$ is obviously symmetric, it is moreover diagonal and non-degenerate with components
\be
 \rho_{AB}=\frac{\lambda u}{\lambda ur-1}\delta_{AB} , \qquad \rho_{IJ}=\frac{1}{r}\delta_{IJ} .
\ee
There are thus two eigenvalue-blocks of dimension $n_\sigma$ and $n_\Sigma$, or $[n_\sigma,n_\Sigma]$. For instance, for $n=7$ we can construct explicit solutions with eigenvalues $\{a,a,a,a,b\}$, for $n=8$ we can have 
$\{a,a,a,a,a,b\}$ and $\{a,a,a,a,b,b\}$ and so on (in all cases $a\neq0\neq b$, $a\neq b$). These provide examples for of the shearing spacetimes of  theorem~\ref{prop_GSHD}  
and proposition~\ref{prop_non_deg2}, although an eigenvalue structure {more general than the one of metric \eqref{general_violating} can in principle exist}. {In particular, the optical constraint is clearly violated (cf. eq.~\eqref{canformL}).} Note also that we cannot construct a six-dimensional example with this method. By setting the cosmological constant to zero in the above metrics ($\lambda=0$) one obtains Ricci-flat spacetimes -- however, these are  direct products and $\rhob$ becomes degenerate, which is of no interest to the present discussion.

By taking direct products one can trivially generate many similar examples with a degenerate $\bl$. These will have $\WDS{ij} \neq 0$, however still with $\WDA{ij} =0$ (see \cite{PraPraOrt07}), in agreement with point (i) of theorem~\ref{prop_GSHD}.

\subsubsection{On possible additional (m)WANDs}

\label{subsubsec_additionalW}

One might wonder whether spacetimes \eqref{general_violating} admit other multiple (double) WANDs (different from $\bl$ and $\bn$, eqs.~\eqref{l_ex} and \eqref{n_ex}) and in particular whether those can obey the optical constraint. Take a generic null vector $\bk=\alpha\pa_u+\beta\pa_r+\gamma^A\bbm_{(A)}+\gamma^I\bbm_{(I)}$ (where the $\bbm_{(A)}$ [$\bbm_{(I)}$] have components only in the subspace of $\d\sigma_\lambda^2$ [$\d\Sigma_0^2$]), i.e., $\alpha^2\lambda r^2+2\alpha\beta+\delta_{AB}\gamma^A\gamma^B+\delta_{IJ}\gamma^I\gamma^J=0$. Since for $\gamma^A=0=\gamma^I$ $\bk$ gives the directions of $\bl$ and $\bn$ (respectively for $\alpha=0$ and $\beta=-\frac{1}{2}\alpha\lambda r^2$), we now assume that $\gamma^A$ and $\gamma^I$ are not both zero (so that $\alpha\neq0$). Using \eqref{C_ur_ex}, \eqref{C_ABIJ_ex} and the Bel-Debever criteria \cite{Ortaggio09} it follows that: (i) $\bk$ is a WAND $\Leftrightarrow C^0_{ABCD}\gamma^A\gamma^C=0=C^0_{IJKL}\gamma^I\gamma^K$  
(so that, using also \eqref{C_ur_ex}, one gets $C_{abcd}k^bk^d=0$); (ii) $\bk$ is an mWAND $\Leftrightarrow C^0_{ABCD}\gamma^A=0=C^0_{IJKL}\gamma^I$ (in which case, using also \eqref{C_ur_ex}, one gets $C_{abcd}k^d=0$). 

These conditions on the Weyl tensors of $\d\sigma_\lambda^2$ and $\d\Sigma_0^2$ are generically not satisfied (see an example in section~\ref{subsubsec_ex2}), so in general spacetimes \eqref{general_violating} do not admit any other WANDs (not even single) apart from $\bl$ and $\bn$. However, it is also clear that in special cases additional mWANDs may exist. For example, when either $\d\sigma_\lambda^2$ or $\d\Sigma_0^2$ are conformally flat (which is necessarily the case if $n_\sigma<4$  or $n_\Sigma<4$) spacetimes \eqref{general_violating} admit a continuous infinity of mWANDs (generically non-geodesic but some can be geodesic, see an example in section~\ref{subsubsec_ex1}).

\subsection{An explicit example with additional mWANDs}

\label{subsubsec_ex1}

For the sake of definiteness, one can for instance construct an explicit solution for $n=7$ by taking (\ref{general_violating}) with 
\beqn
 & & \d\sigma_\lambda^2=V(\rho)\d\tau^2+V^{-1}(\rho)\d\rho^2+\rho^2(\d\theta^2+\sin\theta^2\d\phi^2) , \qquad V(\rho)=1-\frac{\mu}{\rho}-\lambda \rho^2 , \nonumber \\
 & & \d\Sigma_0^2=\d z^2 , \label{NOCmetric}
\eeqn
and a coordinate range such that $V(\rho)>0$.

Taking $\bl=\pa_r$ and the orthonormal vectors
\beqn
 & & \bbm_{(2)}=A^{-1}V^{-1/2}(\rho)\pa_\tau , \qquad \bbm_{(3)}=A^{-1}V^{1/2}(\rho)\pa_\rho, \qquad \bbm_{(4)}=A^{-1}\rho^{-1}\pa_\theta, \nonumber \\ 
 & & \bbm_{(5)}=A^{-1}\rho^{-1}\sin\theta^{-1}\pa_\phi, \qquad \bbm_{(6)}=r^{-1}\pa_z , \qquad A=\lambda ur-1 ,
\eeqn
one finds
\be
 \rho_{22}=\rho_{33}=\rho_{44}=\rho_{55}=\frac{\lambda u}{\lambda ur-1} , \qquad \rho_{66}=\frac{1}{r} .
\ee

The only non-zero Weyl frame components read (recall definition \eqref{def_W})
\be
 W_{23}=W_{45}=\frac{1}{(\lambda ur-1)^2}\frac{\mu}{\rho^3} , \qquad W_{24}=W_{25}=W_{34}=W_{35}=-\frac{1}{(\lambda ur-1)^2}\frac{\mu}{2\rho^3} .
\ee

However, note  that since $\d\Sigma_0^2=\d z^2$ is (conformally) flat, it follows from section~\ref{subsubsec_additionalW} that metric \eqref{general_violating} with \eqref{NOCmetric} admits also other  mWANDs of the form  
$\bk=\frac{1}{r^2}[\pa_u-\frac{1}{2}r^2(\lambda+\gamma^2)\pa_r+\gamma\pa_z]$, where $\gamma$ is an arbitrary function.

It turns out that if $\gamma_{,r}=\gamma_{,u}=\gamma_{,z}=0$ the mWAND $\bk$ is geodesic, in which case it becomes twistfree iff $\gamma$ is a constant and it obeys the optical constraint iff $\gamma^2=|\lambda|$ (so that it is also twistfree), i.e.,
\be
	\bk=\frac{1}{r^2}\left[\partial_u-\frac{1}{2}(\lambda+|\lambda|)r^2\pa_r\pm\sqrt{|\lambda|}\pa_z\right] .
\ee
In this case the corresponding optical matrix is of the form $\{a,a,a,a,a\}$ for $\lambda>0$ and $\{a,a,a,a,0\}$ for $\lambda<0$ (in both cases $a\neq0$).

\subsection{An explicit example without additional mWANDs}

\label{subsubsec_ex2}

{It follows from section~\ref{subsubsec_additionalW} that in order for metric \eqref{general_violating} to admit only two double WANDs we need $n_\sigma\ge4$ and $n_\Sigma\ge4$, i.e., at least 10 spacetime dimensions. If we now take $\d\sigma_\lambda^2$ as in \eqref{NOCmetric} but} $\d\Sigma_0^2$ corresponding to the Riemannian version of 4d Schwarzschild, {by looking at the conditions of section~\ref{subsubsec_additionalW} it is easy to see that indeed} the {\em only} mWANDS are the $\bl$ and $\bn$ discussed above (eqs.~\eqref{l_ex} and \eqref{n_ex}). Such metric thus constitutes an example of an Einstein spacetime with {\em all double WANDs violating the optical constraint}. Recall also that such mWANDs \eqref{l_ex} and \eqref{n_ex} are geodesic and non-degenerate.

%
%

\begin{thebibliography}{10}

\bibitem{GolSac62}
J.~N. Goldberg and R.~K. Sachs.
\newblock A theorem on {P}etrov types.
\newblock {\em Acta Phys. Polon.}, \textnormal{Suppl.} 22:13--23, 1962.

\bibitem{NP}
E.~T. Newman and R.~Penrose.
\newblock An approach to gravitational radiation by a method of spin
  coefficients.
\newblock {\em J. Math. Phys.}, 3:566--578, 1962.
\newblock See also E. Newman and R. Penrose (1963), Errata, {\em J. Math.
  Phys.} 4:998.

\bibitem{Stephanibook}
H.~Stephani, D.~Kramer, M.~MacCallum, C.~Hoenselaers, and E.~Herlt.
\newblock {\em Exact Solutions of {E}instein's Field Equations}.
\newblock Cambridge University Press, Cambridge, second edition, 2003.

\bibitem{penrosebook2}
R.~Penrose and W.~Rindler.
\newblock {\em Spinors and Space-Time}, volume~2.
\newblock Cambridge University Press, Cambridge, 1986.

\bibitem{Kerr63}
R.~P. Kerr.
\newblock Gravitational field of a spinning mass as an example of algebraically
  special metrics.
\newblock {\em Phys. Rev. Lett.}, 11:237--238, 1963.

\bibitem{Coleyetal04}
A.~Coley, R.~Milson, V.~Pravda, and A.~Pravdov\'a.
\newblock Classification of the {W}eyl tensor in higher dimensions.
\newblock {\em Class. Quantum Grav.}, 21:L35--L41, 2004.

\bibitem{OrtPraPra12rev}
M.~Ortaggio, V.~Pravda, and A.~Pravdov\'a.
\newblock Algebraic classification of higher dimensional spacetimes based on
  null alignment.
\newblock {\em Class. Quantum Grav.}, 2012,
\newblock to appear.

\bibitem{Pravdaetal04}
V.~Pravda, A.~Pravdov\'a, A.~Coley, and R.~Milson.
\newblock Bianchi identities in higher dimensions.
\newblock {\em Class. Quantum Grav.}, 21:2873--2897, 2004.
\newblock See also V. Pravda, A. Pravdov\'a, A. Coley and R. Milson {\em Class.
  Quantum Grav.} {\bf 24} (2007) 1691 (corrigendum).

\bibitem{Coleyetal04vsi}
A.~Coley, R.~Milson, V.~Pravda, and A.~Pravdov\'a.
\newblock Vanishing scalar invariant spacetimes in higher dimensions.
\newblock {\em Class. Quantum Grav.}, 21:5519--5542, 2004.

\bibitem{OrtPraPra07}
M.~Ortaggio, V.~Pravda, and A.~Pravdov\'a.
\newblock Ricci identities in higher dimensions.
\newblock {\em Class. Quantum Grav.}, 24:1657--1664, 2007.

\bibitem{Durkeeetal10}
M.~Durkee, V.~Pravda, A.~Pravdov\'a, and H.~S. Reall.
\newblock Generalization of the {G}eroch-{H}eld-{P}enrose formalism to higher
  dimensions.
\newblock {\em Class. Quantum Grav.}, 27:215010, 2010.

\bibitem{Ortaggioetal12}
M.~Ortaggio, V.~Pravda, A.~Pravdov\'a, and H.~S. Reall.
\newblock On a five-dimensional version of the {G}oldberg-{S}achs theorem.
\newblock {\em Class. Quantum Grav.}, 29:205002, 2012.

\bibitem{MyePer86}
R.~C. Myers and M.~J. Perry.
\newblock Black holes in higher dimensional space-times.
\newblock {\em Ann. Phys. (N.Y.)}, 172:304--347, 1986.

\bibitem{FroSto03}
V.~P. Frolov and D.~Stojkovi\'c.
\newblock Particle and light motion in a space-time of a five-dimensional
  rotating black hole.
\newblock {\em Phys. Rev. {\rm D}}, 68:064011, 2003.

\bibitem{PraPraOrt07}
V.~Pravda, A.~Pravdov\'a, and M.~Ortaggio.
\newblock Type {D} {E}instein spacetimes in higher dimensions.
\newblock {\em Class. Quantum Grav.}, 24:4407--4428, 2007.

\bibitem{GodRea09}
M.~Godazgar and H.~S. Reall.
\newblock Algebraically special axisymmetric solutions of the
  higher-dimensional vacuum {E}instein equation.
\newblock {\em Class. Quantum Grav.}, 26:165009, 2009.

\bibitem{Durkee09}
M.~Durkee.
\newblock Type {II} {E}instein spacetimes in higher dimensions.
\newblock {\em Class. Quantum Grav.}, 26:195010, 2009.

\bibitem{DurRea09}
M.~Durkee and H.~S. Reall.
\newblock A higher-dimensional generalization of the geodesic part of the
  {G}oldberg-{S}achs theorem.
\newblock {\em Class. Quantum Grav.}, 26:245005, 2009.

\bibitem{Taghavi-Chabert11}
A.~Taghavi-Chabert.
\newblock Optical structures, algebraically special spacetimes, and the
  {G}oldberg-{S}achs theorem in five dimensions.
\newblock {\em Class. Quantum Grav.}, 28:145010, 2011.

\bibitem{Wylleman_priv}
L.~Wylleman.
\newblock The optical constraint and the {G}oldberg-{S}achs theorem in five and
  higher dimensions,
\newblock to appear.

\bibitem{OrtPraPra09}
M.~Ortaggio, V.~Pravda, and A.~Pravdov\'a.
\newblock Higher dimensional {K}err-{S}child spacetimes.
\newblock {\em Class. Quantum Grav.}, 26:025008, 2009.

\bibitem{MalPra11}
T.~M\'alek and V.~Pravda.
\newblock {K}err-{S}child spacetimes with ({A})d{S} background.
\newblock {\em Class. Quantum Grav.}, 28:125011, 2011.

\bibitem{OrtPraPra09b}
M.~Ortaggio, V.~Pravda, and A.~Pravdov\'a.
\newblock Asymptotically flat, algebraically special spacetimes in higher
  dimensions.
\newblock {\em Phys. Rev. {\rm D}}, 80:084041, 2009.

\bibitem{OrtPraPra10}
M.~Ortaggio, V.~Pravda, and A.~Pravdov\'a.
\newblock Type {III} and {N} {E}instein spacetimes in higher dimensions:
  general properties.
\newblock {\em Phys. Rev. {\rm D}}, 82:064043, 2010.

\bibitem{PodOrt06}
J.~Podolsk\'y and M.~Ortaggio.
\newblock {R}obinson-{T}rautman spacetimes in higher dimensions.
\newblock {\em Class. Quantum Grav.}, 23:5785--5797, 2006.

\bibitem{PodZof09}
J.~Podolsk\'y and M.~\v{Z}ofka.
\newblock General {K}undt spacetimes in higher dimensions.
\newblock {\em Class. Quantum Grav.}, 26:105008, 2009.

\bibitem{Taghavi-Chabert11b}
A.~Taghavi-Chabert.
\newblock The complex {G}oldberg-{S}achs theorem in higher dimensions.
\newblock {\em J. Geom. Phys.}, 62:981--1012, 2012.

\bibitem{PraPra08}
A.~Pravdov\'a and V.~Pravda.
\newblock The {N}ewman-{P}enrose formalism in higher dimensions: vacuum
  spacetimes with a non-twisting geodetic multiple {W}eyl aligned null
  direction.
\newblock {\em Class. Quantum Grav.}, 25:235008, 2008.

\bibitem{Brinkmann25}
H.~W. Brinkmann.
\newblock Einstein spaces which are mapped conformally on each other.
\newblock {\em Math. Ann.}, 94:119--145, 1925.

\bibitem{OrtPraPra11}
M.~Ortaggio, V.~Pravda, and A.~Pravdov\'a.
\newblock On higher dimensional {E}instein spacetimes with a warped extra
  dimension.
\newblock {\em Class. Quantum Grav.}, 28:105006, 2011.

\bibitem{Trautman02b}
A.~Trautman.
\newblock Robinson manifolds and the shear-free condition.
\newblock {\em Int. J. Mod. Phys. {\rm A}}, 17:2735--2737, 2002.

\bibitem{EmpRea02prl}
R.~Emparan and H.~S. Reall.
\newblock A rotating black ring solution in five dimensions.
\newblock {\em Phys. Rev. Lett.}, 88:101101, 2002.

\bibitem{ElvEmp03}
H.~Elvang and R.~Emparan.
\newblock Black rings, supertubes, and a stringy resolution of black hole
  non-uniqueness.
\newblock {\em JHEP}, 11:035, 2003.

\bibitem{PlebHac75}
J.~F. Pleba\'nski and S.~Hacyan.
\newblock Null geodesic surfaces and the {G}oldberg-{S}achs theorem in complex
  {R}iemannian spaces.
\newblock {\em J. Math. Phys.}, 16:2403--2407, 1975.

\bibitem{RobTra83}
I.~Robinson and A.~Trautman.
\newblock Conformal geometry of flows in $n$ dimensions.
\newblock {\em J. Math. Phys.}, 24:1425--1429, 1983.

\bibitem{NurTra02}
P.~Nurowski and A.~Trautman.
\newblock Robinson manifolds as the {L}orentzian analogs of {H}ermite
  manifolds.
\newblock {\em Differential Geom. Appl.}, 17:175--195, 2002.

\bibitem{Trautman02a}
A.~Trautman.
\newblock Robinson manifolds and {C}auchy-{R}iemann spaces.
\newblock {\em Class. Quantum Grav.}, 19:R1--R10, 2002.

\bibitem{MasTag08}
L.~J. Mason and A.~Taghavi-Chabert.
\newblock Killing-{Y}ano tensors and multi-hermitian structures.
\newblock {\em J. Geom. Phys.}, 60:907--923, 2010.

\bibitem{robinsonnull}
I.~Robinson.
\newblock Null electromagnetic fields.
\newblock {\em J. Math. Phys.}, 2:290, 1961.

\bibitem{Ortaggio07}
M.~Ortaggio.
\newblock Higher dimensional spacetimes with a geodesic, shearfree, twistfree
  and expanding null congruence.
\newblock Proceedings of the XVII SIGRAV Conference (Turin, 2006) [gr-qc/0701036].

\bibitem{HugMas88}
L.~P. Hughston and L.~J. Mason.
\newblock A generalized {K}err-{R}obinson theorem.
\newblock {\em Class. Quantum Grav.}, 5:275--285, 1988.

\bibitem{Coleyetal03}
A.~Coley, R.~Milson, N.~Pelavas, V.~Pravda, A.~Pravdov\'a, and R.~Zalaletdinov.
\newblock Generalizations of \pp-wave spacetimes in higher dimensions.
\newblock {\em Phys. Rev. {\rm D}}, 67:104020, 2003.

\bibitem{ColHerPel06}
A.~Coley, S.~Hervik, and N.~Pelavas.
\newblock On spacetimes with constant scalar invariants.
\newblock {\em Class. Quantum Grav.}, 23:3053--3074, 2006.

\bibitem{ManSte04}
R.~Mann and C.~Stelea.
\newblock Nuttier {(A)dS} black holes in higher dimensions.
\newblock {\em Class. Quantum Grav.}, 21:2937--2961, 2004.

\bibitem{AwaCha02}
A.~M. Awald and A.~Chamblin.
\newblock A bestiary of higher-dimensional {T}aub-{NUT}-{AdS} spacetimes.
\newblock {\em Class. Quantum Grav.}, 19:2051--2061, 2002.

\bibitem{petrov}
A.~Z. Petrov.
\newblock {\em Einstein Spaces}.
\newblock Pergamon Press, Oxford, translation of the 1961 {R}ussian edition,
  1969.

\bibitem{PraPra05}
V.~Pravda and A.~Pravdov\'a.
\newblock {WAND}s of the black ring.
\newblock {\em Gen. Rel. Grav.}, 37:1277--1287, 2005.

\bibitem{Ortaggio09}
M.~Ortaggio.
\newblock {B}el-{D}ebever criteria for the classification of the {W}eyl tensor
  in higher dimensions.
\newblock {\em Class. Quantum Grav.}, 26:195015, 2009.

\end{thebibliography}

\end{document}